\documentclass[twcolumn,journal,final]{IEEEtran}

\usepackage{graphicx}
\usepackage{bm}
\usepackage{url}
\usepackage{amsfonts,amssymb,amsmath}
\usepackage{mystyle}
\allowdisplaybreaks                              
\centerfigcaptionstrue

\def\Expt{\mathbb E}

\def\Hnc{\bH_1}
\def\Xn{\bX_1}
\def\Yn{\mathbf{Y}_1}
\def\Rnc {R_1}
\def\Wn{\mathbf{W}_1}
\def\Tx{\bX}
\def\Hc{\bH_2}
\def\Xc{\bX_2}
\def\Yc{\mathbf{Y}_2}
\def\Rc{R_2}
\def\Tc{T}
\def\Nn{N_1}
\def\Nc{N_2}
\def\NTx{M}
\def\tr{\text{tr}}
\def\Heq1{{\widetilde{\bH}_1}}
\def\Wc{\mathbf{W}_2}

\begin{document}

\title{Product Superposition for MIMO Broadcast Channels} \author{Yang
  Li, {\em Student Member, IEEE}, and Aria Nosratinia, {\em Fellow,
    IEEE}\thanks{Manuscript received November 11, 2011; revised June
    6, 2012; accepted June 7, 2012. The material in this paper was
    presented in part in ISIT 2011 and 2012.} \thanks{ The authors are
    with the University of Texas at Dallas, Richardson, TX 75080, USA,
    Email: yang@utdallas.edu, aria@utdallas.edu. }\thanks{Communicated
    by Angel Lozano, Editor for Communications.}}

\maketitle

\begin{abstract}
This paper considers the multiantenna broadcast channel without
transmit-side channel state information (CSIT). For this channel, it
has been known that when all receivers have channel state information
(CSIR), the degrees of freedom (DoF) cannot be improved beyond what is
available via TDMA. The same is true if none of the receivers possess
CSIR.  This paper shows that an entirely new scenario emerges when
receivers have unequal CSIR. In particular, orthogonal transmission is
no longer DoF-optimal when one receiver has CSIR and the other does
not.  A multiplicative superposition is proposed for this scenario and
shown to attain the optimal degrees of freedom under a wide set of
antenna configurations and coherence lengths.  Two signaling schemes
are constructed based on the multiplicative superposition. In the
first method, the messages of the two receivers are carried in the row and
column spaces of a matrix, respectively. This method works better than
orthogonal transmission while reception at each receiver is still
interference-free. The second method uses coherent signaling for the
receiver with CSIR, and Grassmannian signaling for the receiver
without CSIR. This second method requires interference cancellation at
the receiver with CSIR, but achieves higher DoF than the first method.
\end{abstract}

\section{Introduction}
\label{sec:Introduction}

In the MIMO broadcast channel, 
when channel state information is available at the receiver (CSIR) but
not at the transmitter (CSIT), orthogonal transmission (e.g., TDMA)
achieves optimal degrees of freedom (DoF)~\cite{Caire2003, Huang2009}.
With neither CSIT nor CSIR, again orthogonal transmission achieves the
best possible DoF~\cite{Jafar2005}.  This paper studies the broadcast
channel where one receiver has full CSIR and another has no CSIR. In
this case, new DoF gains are discovered that can be unlocked with
novel signaling strategies.

The study of broadcast channels with unequal CSIR is motivated by
downlink scenarios where users have different mobilities. Subject to a
fixed downlink pilot transmission schedule, the low-mobility users
have the opportunity to reliably estimate their channels, while the
high-mobility users may not have the same opportunity. 


The main result of this paper is that when one receiver has full CSIR
and the other has none, the achieved DoF is strictly better than that
obtained by orthogonal transmission. For the unequal CSIR scenario, we
propose a {\em product superposition}, where the signals of the two
receivers are multiplied to produce the broadcast signal.
In the following the receiver with full CSIR is referred to as the
{\em static receiver} and the receiver with no CSIR as the {\em
  dynamic receiver}. Two classes of product superposition signaling
are proposed:
\begin{itemize}
\item In the first method, information for both receivers is conveyed
  by the row and column spaces of a transmit signal matrix,
  respectively. The signal matrix is constructed from a product of two
  signals that lie on different Grassmannians. The two receivers do
  {\em not} interfere with each other even though there is no CSIT, a
  main point of departure from traditional superposition
  broadcasting~\cite{Caire2003,Coverbook}.

\item In the second method, information for the static receiver is
  carried by the signal matrix values (coherent signaling), while
  information for the dynamic receiver is transported on the
  Grassmannian. The static receiver is required to decode and cancel
  interference, therefore this method is slightly more involved, but
  it achieves higher DoF compared with the first method.
\end{itemize}
Using the proposed methods, the exact DoF region is found when $\Nn\le
\Nc \le \NTx$, $\Tc\ge 2\Nn$, where $\Nn$, $\Nc$ and $\NTx$ are the
number of antennas at the dynamic receiver, static receiver and
transmitter, respectively, and $\Tc$ is the channel coherence time of
the dynamic receiver. For $\Nc<\Nn\le\NTx$, $\Tc\ge 2\Nn$, we
partially characterize the DoF region when either the channel is the
more capable type~\cite{Gamal1979}, or when the message set is
degraded~\cite{Korner1977}.

We use the following notation throughout the paper: for a matrix
$\bA$, the transpose is denoted with $\bA^t$, the conjugate transpose
with $\bA^\dag$, and the element in row $i$ and column $j$ with
$[\bA]_{i,j}$.  The $k\times k$ identity matrix is denoted with
$\bI_k$. The set of $n\times m$ complex matrices is denoted with
$\mathcal{C}^{n\times m}$.

The organization of this paper is as follows. In
Section~\ref{sec:Preliminaries} we introduce the system model and
preliminary results. Two signaling methods are proposed and studied in
Section~\ref{sec:Broadcasting1} and Section~\ref{sec:Broadcasting2},
respectively.

\section{System Model and Preliminaries}
\label{sec:Preliminaries}

We consider a broadcast channel with an $\NTx$-antenna transmitter and
two receivers. One receiver has access to channel state information
(CSI), and is referred to as the {\em static receiver.}  The other
receiver has no CSI, e.g. due to mobility, and is referred to as the
{\em dynamic receiver}.  The dynamic receiver has $\Nn$ antennas and
the static receiver has $\Nc$ antennas. Denote the channel coefficient
matrices from the transmitter to the dynamic and static receivers by
$\Hnc\in\mathcal{C}^{\Nn\times \NTx}$ and
$\Hc\in\mathcal{C}^{\Nc\times \NTx}$, respectively. We assume that
$\Hnc$ is constant for $\Tc$ symbols (block-fading) and is unknown to
both receivers, while $\Hc$ is known by the static receiver but not
known by the dynamic receiver.\footnote{In practice $\Hc$ for a static
  receiver  may vary across
  intervals of length much greater than $T$. However, for the purposes
  of this paper, once $\Hc$ is assumed to be known to the static
  receiver, its time variation (or lack thereof) does not play any
  role in the subsequent mathematical developments. Therefore in the
  interest of elegance and for a minimal description of the
  requirements for the results, we only state that $\Hc$ is known.}
Neither $\Hnc$ nor $\Hc$ is known by the transmitter
(no CSIT). 

\begin{figure}
\centering
\includegraphics[width=2.5in]{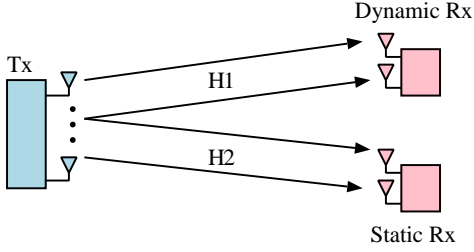}
\caption{Channel model.}
\label{fig:Channel_Model}
\end{figure}

Over $\Tc$ time-slots (symbols) the transmitter sends
$\Tx=[\mathbf{x}_1,\cdots,\mathbf{x}_M]^t$ across $\NTx$ antennas,
where $\mathbf{x}_i\in \mathcal{C}^{\Tc\times 1}$ is the signal vector
sent by the antenna $i$. The normalized signal at the dynamic and
static receivers is respectively
\begin{align}
\Yn &= \Hnc \Tx + \frac{1}{\sqrt{\rho}} \Wn, \nonumber \\
\Yc & = \Hc \Tx + \frac{1}{\sqrt{\rho}} \Wc, \label{eq:channelmodel}
\end{align}
where $\Wn\in \mathcal{C}^{\Nn \times \Tc}$ and $\Wc\in
\mathcal{C}^{\Nc \times \Tc}$ are additive noise with i.i.d. entries
$\mathcal{CN}(0,1)$. Each row of $\Yn\in \mathcal{C}^{\Nn \times \Tc}$
(or $\Yc\in \mathcal{C}^{\Nc \times \Tc}$) corresponds to the received
signal at an antenna of the dynamic receiver (or the static receiver)
over $\Tc$ time-slots.  The transmitter is assumed to have an average
power constraint $\rho$, and therefore, in the normalized channel
model given by~\eqref{eq:channelmodel}, the average power constraint
is:
\begin{equation}
\Expt\big[\sum_{i=1}^{\NTx}\tr(\mathbf{x}_i\mathbf{x}_i^{\dag})\big] =
\Tc. \label{eq:powerconstraint}
\end{equation}

The channel $\Hnc$ has i.i.d. entries with zero mean and unit
variance, but we do {\em not} assign any specific distribution for
$\Hnc$. This general model includes Rayleigh fading as a special case
where the entries of $\Hnc$ are i.i.d. $\mathcal{CN}(0,1)$. The
channel $\Hc$ is assumed to have full rank; this assumption, e.g.,
holds with probability $1$ if the entries of $\Hc$ are drawn
independently according to a continuous distribution.  We focus on the
case of $\NTx = \max(\Nn,\Nc)$ and $\Tc\ge 2\Nn$, which is motivated
by the fact that having more transmit antennas does not increase the
multiplexing gain for either receiver, and the fact that if $\Tc <
2\Nn$, some of the antennas of the dynamic receiver can be deactivated
without any loss in the degrees of freedom (DoF)~\cite{Zheng2002}.

The degrees of freedom at receiver $i$ is defined as:
\begin{equation}
d_i=\lim_{\rho\rightarrow\infty}\frac{R_i(\rho)}{\log \rho}, \nonumber
\end{equation}
where $R_i(\rho)$ is the rate of the dynamic receiver for $i=1$ and
the static receiver for $i=2$, respectively.  

\subsection{Definitions}

\begin{definition}[Isotropically Distributed Matrix~\cite{Marzetta1999}]
\label{def:id}
A random matrix $\bX\in\mathcal{C}^{k\times n}$, where $n\ge k$, is
called isotropically distributed (i.d.) if its distribution is
invariant under unitary transformations, i.e., for any deterministic
$n\times n$ unitary matrix $\mathbf{\Phi}$,
\begin{equation}
p(\bX) = p( \bX \mathbf{\Phi}).
\end{equation}
\end{definition}
An example of i.d. matrices is $\bX$ with i.i.d. $\mathcal{CN}(0,1)$
entries.
\begin{remark}
\label{remark:1}
An interesting property of i.d. matrices is that if $\bX$ is i.d. and
$\mathbf{\Phi}$ is a random unitary matrix that is independent of
$\bX$, then $\mathbf{ X\Phi}$ is independent of
$\mathbf{\Phi}$~\cite[Lemma 4]{Zheng2002}.  That is, any rotation to
an i.d. matrix is essentially ``invisible.''
\end{remark}

\begin{definition}[Stiefel manifold~\cite{Boothbybook}]
\label{def:stiefel}
The Stiefel manifold $\mathbb{F}(n,k)$, where $n>k$, is the set of
all $k\times n$ unitary matrices, i.e.,
\begin{equation}
\mathbb{F}(n,k) =\big\{ \bQ\in \mathcal{C}^{k\times n}:\, \mathbf{Q Q}^{\dag} =
\bI_k \big\}. \nonumber
\end{equation}
\end{definition}

For $k=1$, the manifold $\mathbb{F}(n,1)$ is the collection of all
$n$-dimensional vectors with unit norm, i.e., the surface of a unit
ball.

\begin{definition}[Grassmann manifold~\cite{Boothbybook}]
The Grassmann manifold $\mathbb{G}(n,k)$, where $n>k$, is the set of
all $k$-dimensional subspaces of $\mathcal{C}^n$.
\end{definition}
\begin{remark}
\label{remark:dimGrassmann}
The (complex) dimension of
$\mathbb{G}(n,k)$ is
\begin{equation}
\dim\big(\mathbb{G}(n,k)\big) = k(n-k),
\end{equation}
i.e., each point in $\mathbb{G}(n,k)$ has a neighborhood that is
equivalent (homeomorphic) to a ball in the Euclidean space of complex
dimension $k(n-k)$. The dimensionality of Grassmannian can also be
viewed as follows. For any matrix $\bQ$, there exists a $k\times k$
full rank matrix $\mathbf{U}$ so that
\begin{align}
\bQ^*=\mathbf{U Q} = \begin{bmatrix}
&1\ &\cdots\ &0 \ &x_{1,k+1}\ &\cdots\ &x_{1n} \\
&0\ &\cdots\ &0 \\
&\vdots &\quad &\vdots &\quad &\quad &\vdots \\
&0 &\cdots &1 & x_{k,k+1} &\cdots & x_{kn}
\end{bmatrix}, \label{eq:GrassmannianEquivalence}
\end{align}
where $\bQ$ and $\bQ^*$ span the same row
space. Therefore, each point in $\mathbb{G}(n,k)$ is determined by
$k(n-k)$ complex parameters $x_{ji}$, for $1\le j\le k$ and $k+1\le
i\le n$. In other words, a $k$-dimension subspace in $\mathcal{C}^n$
is uniquely decided by $k(n-k)$ complex variables.
\end{remark}

\subsection{Non-coherent Point-to-point Channels}
\label{sec:noncoherent}

The analysis in this paper uses insights and results from non-coherent
communication in point-to-point MIMO channels, which are briefly
outlined below.

\subsubsection{Intuition}
Consider a point-to-point $M\times N$ MIMO channel where the receiver
does not know the channel $\bH$, namely a non-coherent channel.

At high SNR the additive noise is negligible, so the received signal
$\mathbf{Y} \approx \bH \Tx$, where $\Tx$ is the transmitted
signal. Because $\Tx$ is multiplied by a random and unknown
$\bH$, the receiver cannot decode
$\Tx$. 
However, communication is still possible because, for any non-singular
$\bH$, the received signal $\mathbf{Y}$ spans the same row space as
$\Tx$. Therefore, the row space of $\Tx$ can be used to carry
information without the need to know $\bH$, i.e., the codebook
consists of matrices with different row spaces.

Conveying information via subspaces can be viewed as communication on
the Grassmann manifold where each distinct point in the manifold
represents a different subspace~\cite{Zheng2002}. In this case, the
codewords (information) are represented by subspaces, which differs
from the coherent communication that maps each codeword into one point
in a Euclidean space~\cite{Proakisbook}. Intuitively, the information
of a Grassmannian codeword is carried by $k(n-k)$ variables, as seen
in~\eqref{eq:GrassmannianEquivalence}.

\subsubsection{Optimal Signaling}
The design of an optimal signaling can be viewed as sphere packing
over Grassmannians~\cite{Zheng2002}. At high SNR, the optimal signals
are isotropically distributed unitary matrices~\cite{ Marzetta1999,
  Zheng2002}. In addition, the optimal number of transmit antennas
depends on the channel coherence time. For a short coherence interval,
using fewer antennas may lead to a higher capacity, and vice
versa. The optimal number of transmit antennas is
\begin{equation}
K = \min (M,N,\lfloor \Tc/2 \rfloor), \label{eq:K}
\end{equation}
where $T$ is the channel coherence time, i.e., the number of symbols
that the channel remains constant. Therefore, the optimal signals are
$K\times \Tc$ unitary matrices. In other words, $K$ antennas ($K\le M$)
are in use and they transmit equal-energy and mutually orthogonal
vectors. These unitary matrices reside in $\mathbb{G}(\Tc,K)$ and each
is interpreted as a representation of the subspace it spans.  This
method achieves the {\em maximum} DoF $K(T-K)$ over $T$
time-slots. Note that the DoF coincides with the
dimensionality of the Grassmannian $\mathbb{G}(\Tc,K)$.

\subsubsection{Subspace Decoding}
Unlike coherent communication, in non-coherent signaling the
information is embedded in the subspaces instead of the signal
values. As long as two matrices span the same subspace, they
correspond to the same message.  Maximum-likelihood decoding chooses
the codeword whose corresponding subspace is the closest one to the
subspace spanned by the received signal. For example
in~\cite{Hochwald2000}, the received signals are projected on the
subspaces spanned by different codewords, and then the one is chosen
with the maximum projection energy. More precisely, for the
transmitted signals $\bX_i\in \mathcal{C}^{K\times \Tc}$ from a unitary
codebook $\mathcal{X}$, and the received signals $\mathbf{Y} \in
\mathcal{C}^{K\times \Tc}$, the ML detector is
\begin{equation}
\hat{\bX}_{ML} = \arg \max_{\bX_i\in \mathcal{X}}
tr\{\mathbf{Y} \bX_i^{\dag}\bX_i\mathbf{Y}^{\dag}\}. \label{eq:MLDecoder}
\end{equation}

\subsection{A Baseline Scheme: Orthogonal Transmission}
\label{sec:baseline}

For the purposes of establishing a baseline for comparison, we begin
by considering a time-sharing (orthogonal transmission) that acquires
CSIR via training in each interval and uses Gaussian signaling.  This
baseline method has been chosen to highlight the differences of the
heterogeneous MIMO broadcast channel of this paper with two other
known scenarios: It is known that for a broadcast channel with no CSIT and
perfect CSIR, orthogonal transmission ahieves the optimal DoF
region~\cite{Huang2009}. Also, a training-based method with Gaussian
signaling is sufficient to achieve DoF optimality~\cite{Zheng2002} for
the point-to-point noncoherent MIMO channel\footnote{Grassmannian
  signaling is superior, but the same slope of the rate vs. SNR curve
  is obtained with training and Gaussian signaling in the
  point-to-point MIMO channel.}.

In orthogonal transmission, the transmitter communicates with the two
receivers in a time-sharing manner.  When transmitting to the dynamic
receiver, it is optimal if the transmitter activates only $K$ out of
$\NTx$ antennas: it sends pilots from the $K$ antennas sequentially
over the first $K$ time-slots; the dynamic receiver estimates the
channel by using, e.g., minimum-mean-square-error (MMSE)
estimation. Then, the transmitter sends data during the remaining
$(\Tc-K)$ time-slots, and the dynamic receiver decodes the data by
using the estimated channel
coefficients~\cite{Zheng2002,Hassibi2003}. Using this strategy, the
maximum rate achieved by the dynamic receiver is:
\begin{equation}
K(1-\frac{K}{\Tc}) \log\rho + O(1). \label{eq:OTD}
\end{equation}
The operating point in the achievable DoF region where the transmitter
communicates exclusively with the dynamic receiver is denoted with
$\mathcal{D}_1$.
\begin{equation}
\mathcal{D}_1 = \big(K(1-\frac{K}{T}),\, 0 \big). \label{eq:OTS}
\end{equation}

For the static receiver the channel is assumed to be known at the
receiver, therefore data is transmitted to it coherently. 
The maximum rate achieved
by the static receiver is~\cite{Telatar1999}
\begin{equation}
\min(\NTx,\Nc) \log \rho + O(1).
\end{equation}
The operating point in the DoF
region where the transmitter communicates only with the static
receiver is denoted with $\mathcal{D}_2$.
\begin{equation}
\mathcal{D}_2 = \big(0,\, \min(\NTx,\Nc) \big).
\end{equation}

Time-sharing between the two points of $\mathcal{D}_1$ and
$\mathcal{D}_2$ yields the achievable DoF region
\begin{equation}
\bigg( tK(1-\frac{K}{T}),\, (1-t)\min(\NTx,\Nc)
\bigg), \label{eq:dofTDMA}
\end{equation} 
where $t$ is a time-sharing variable.

\section{Grassmannian Superposition for Broadcast Channel}
\label{sec:Broadcasting1}

In this section, we propose a signaling method that attains DoF region
superior to orthogonal transmission, and allows each receiver to
decode its message while being oblivious of the other receiver's message.

\subsection{A Toy Example}
\label{sec:toy1}

Consider $\NTx=\Nc=2$, $\Nn=1$ and $\Tc=2$.
From Section~\ref{sec:baseline}, orthogonal transmission attains
$1/2$ DoF per time-slot for the dynamic receiver and $2$ DoF per
time-slot for the static receiver. By time-sharing between the two
receivers, the following DoF region is achieved
\begin{equation}
(\frac{t}{2},\ 2-2t),
\end{equation}
where $t\in [0,1]$ is a time-sharing parameter.

We now consider the transmitter sends a product of signal vectors over
$2$ time-slots
\begin{equation}
\Tx= \mathbf{x}_2 \mathbf{x}_1^t\in \mathcal{C}^{2\times 2},
\end{equation}
where $\mathbf{x}_1=[x_1^{(1)} \ x_2^{(1)} ]^t$ and $\mathbf{x}_2 =
[x_1^{(2)}\ x_2^{(2)}]^t$ are the signals for the dynamic receiver and
the static receiver, respectively. The vectors $\mathbf{x}_1$ and
$\mathbf{x}_2$ have unit-norm and from codebooks that lie on
$\mathbb{G}(2,1)$.

The signal at the dynamic receiver is 
\begin{align}
\mathbf{y}_1& = [h_1^{(1)}\ h_2^{(1)}] \begin{bmatrix}x_1^{(2)}\\ x_2^{(2)}\end{bmatrix}
[x_1^{(1)}\ x_2^{(1)}] + \frac{1}{\sqrt{\rho}}[w_1^{(1)}\ w_2^{(1)}] \nonumber\\ 
& = \tilde{h}^{(1)}\,[x_1^{(1)}\ x_2^{(1)}] +
\frac{1}{\sqrt{\rho}}[w_1^{(1)}\ w_2^{(1)}], \label{eq:toy1rx1}
\end{align}
where $[h_1^{(1)}, h_2^{(1)}]$ is the isotropically distributed
channel vector, and $\tilde{h}^{(1)}$ is the equivalent channel
coefficient seen by the dynamic receiver.

The subspace spanned by $\mathbf{x}_1^t$ is the same as
$\tilde{h}^{(1)}\mathbf{x}_1^t$, so at high SNR the dynamic receiver
is able to determine the direction specified by $\mathbf{x}_1^t$.
From Section~\ref{sec:noncoherent}, the dynamic receiver attains $1/2$
DoF per time-slot, which is optimal even in the absence of the static
receiver.

Consider the signal of the static receiver at time-slot $1$:
\begin{equation}
\mathbf{y}_2=
\Hc \begin{bmatrix}x_1^{(2)}\\ x_2^{(2)}\end{bmatrix}x_1^{(1)}
+ \frac{1}{\sqrt{\rho}} \begin{bmatrix}w_1^{(2)}\\ w_2^{(2)}\end{bmatrix}.
\end{equation}
Because the static receiver knows $\Hc$, it can invert the
channel\footnote{The noise enhancement induced by channel inversion
  will not affect the DoF of the static receiver.} as
long as $\Hc$ is non-singular:
\begin{equation}
\big(\Hc^{-1}\mathbf{y}_2\big)^t = x_1^{(1)} [x_1^{(2)} \ x_2^{(2)}] +
    [w_1^{(2)}\ w_2^{(2)}]\Hc^{-t}.
\end{equation}
The equivalent (unknown) channel seen by the static receiver is
$x_1^{(1)}$, i.e., part of the dynamic receiver's signal. Using
Grassmannian signaling via the subspace of $\mathbf{x}_2$, the DoF
achieved is again $1/2$ per time-slot.

Time-sharing between the proposed scheme and $\mathcal{D}_2$ (transmitting
only to the static receiver) yields the achievable DoF region
\begin{equation}
\big(\frac{1}{2}t,\ 2-\frac{3}{2}t\big).
\end{equation}
The above region is strictly larger than that of orthogonal
transmission, as shown in Figure~\ref{fig:dofToy1}.  The static
receiver achieves $1/2$ DoF ``for free'' in the sense that this DoF
was extracted for the static receiver without reducing the dynamic
receiver's DoF.
\begin{figure}
\centering
\includegraphics[width=3.5in]{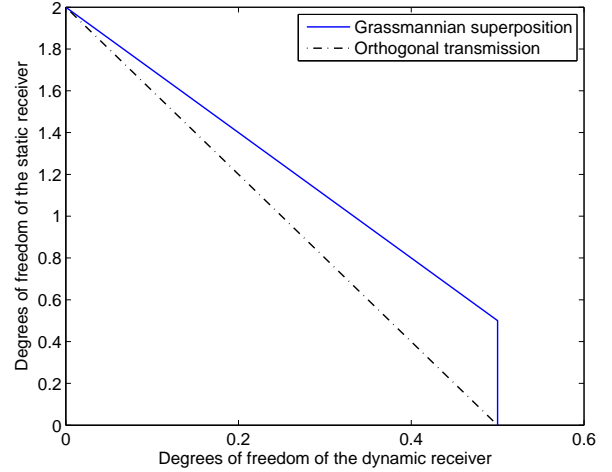}
\caption{DoF region of the toy example 1.}
\label{fig:dofToy1}
\end{figure}

\subsection{Grassmannian Superposition Signaling}

Based on the previous example, we design a general signaling method
(the Grassmannian superposition) with two properties: (1)~information
is carried by subspaces and (2)~two signal matrices are superimposed
multiplicatively so that their row (or column) space is unaffected by
multiplying the other receiver's signal matrix. Two separate cases are
considered based on whether the number of static receiver antennas is
larger than the number of dynamic receiver antennas.

\subsubsection{$\Nn<\Nc$}
\label{section:NnLessNc}
The transmitter sends $\Tx\in\mathcal{C}^{\Nc\times \Tc}$ across
$\NTx=\Nc$ antennas over an interval of length $\Tc$:
\begin{equation}
\Tx= \sqrt{\frac{\Tc}{\Nn}} \Xc \Xn, \label{eq:Tx1}
\end{equation}
where $\Xn\in \mathcal{C}^{\Nn\times \Tc}$ and $\Xc\in
\mathcal{C}^{\Nc\times \Nn}$ are the signals for the dynamic receiver
and the static receiver, respectively. Here, $\sqrt{\Tc/\Nn}$ is a
normalizing factor to satisfy the power
constraint~\eqref{eq:powerconstraint}. Information for both receivers
are sent over the Grassmannian, namely $\Xn$ is from a codebook
$\mathcal{X}_1\subset \mathbb{G}(\Tc,\Nn)$ and $\Xc$ is from a
codebook $\mathcal{X}_2\subset \mathbb{G}(\Nc,\Nn)$. The codebook
$\mathcal{X}_1$ and $\mathcal{X}_2$ are chosen to be isotropically
distributed unitary matrices (see Section~\ref{sec:codebook1} for more
details).

A sketch of the argument for the DoF achieved by the Grassmannian
superposition is as follows. The noise is negligible at high SNR, so
the signal at the dynamic receiver is approximately
\begin{equation}
\Yn \approx
\sqrt{\frac{\Tc}{\Nn}} \Hnc \Xc \Xn \in \mathcal{C}^{\Nn\times \Tc}. \label{eq:Y_1Intutitive}
\end{equation}
The row space of $\Xn$ can be determined based on $\Yn$, and then
$(\Tc-\Nn)\Nn$ independent variables (DoF) that specify
the row space are recovered, i.e., the transmitted point $\Xn$ in
$\mathcal{X}_1\in \mathbb{G}(\Tc,\Nn)$ is found.

For the static receiver, since $\Hc$ is known by the receiver, it
inverts the channel (given that $\Hc$ is non-singular)
\begin{equation}
\Hc^{-1}\Yc \approx \sqrt{\frac{\Tc}{\Nn}} \Xc \Xn\in
\mathcal{C}^{\Nc\times \Tc},
\end{equation}
which has approximately the same column space as $\Xc$. The
transmitted point $\Xc$ in $\mathcal{X}_2\in \mathbb{G}(\Nc,\Nn)$ will
be recovered from the column space of $\Hc^{-1}\Yc$, producing
$(\Nc-\Nn)\Nn$ DoF.

Therefore, the proposed scheme attains the DoF pair
\begin{equation}
\mathcal{D}_3 = \bigg(\Nn(1-\frac{\Nn}{\Tc}),\,  \frac{\Nn}{\Tc}(\Nc-\Nn)  \bigg). 
\end{equation}
The result is more formally stated as follows:
\begin{theorem}[$\Nn<\Nc$]
\label{thm:G-G}
Consider a broadcast channel with an $M$-antenna transmitter, a
dynamic receiver and a static receiver with $\Nn$ and $\Nc$ antennas,
respectively, with coherence time $\Tc$ for the dynamic channel. The
Grassmannian superposition achieves the rate pair
\begin{align*}
\begin{cases}
\Rnc = \Nn\big(1-\frac{\Nn}{\Tc}\big)\log\rho + O(1)
\\ \Rc = \frac{\Nn}{\Tc} (\Nc-\Nn)
\log\rho + O(1)
\end{cases}.
\end{align*}
The corresponding DoF pair is denoted $$\mathcal{D}_3=\bigg(\Nn(1-\frac{\Nn}{\Tc}),\,  \frac{\Nn}{\Tc}(\Nc-\Nn) \bigg).$$ If we denote
the DoF for the single-user operating points for the dynamic and
static user with $\mathcal{D}_1$, $\mathcal{D}_2$ respectively, the
achievable DoF region consists of the convex hull of $\mathcal{D}_1$,
$\mathcal{D}_2$ and $\mathcal{D}_3$.
\end{theorem}
\begin{proof}
See Appendix~\ref{app:thmdof1}.
\end{proof}

From Theorem~\ref{thm:G-G} the static receiver attains a ``free''
rate of
\begin{equation}
 \Delta R_1 = \frac{\Nn}{\Tc} (\Nc-\Nn) \log\rho + O(1). \label{eq:R_2Gain1}
\end{equation}

\begin{figure}
\centering
\includegraphics[width=3.5in]{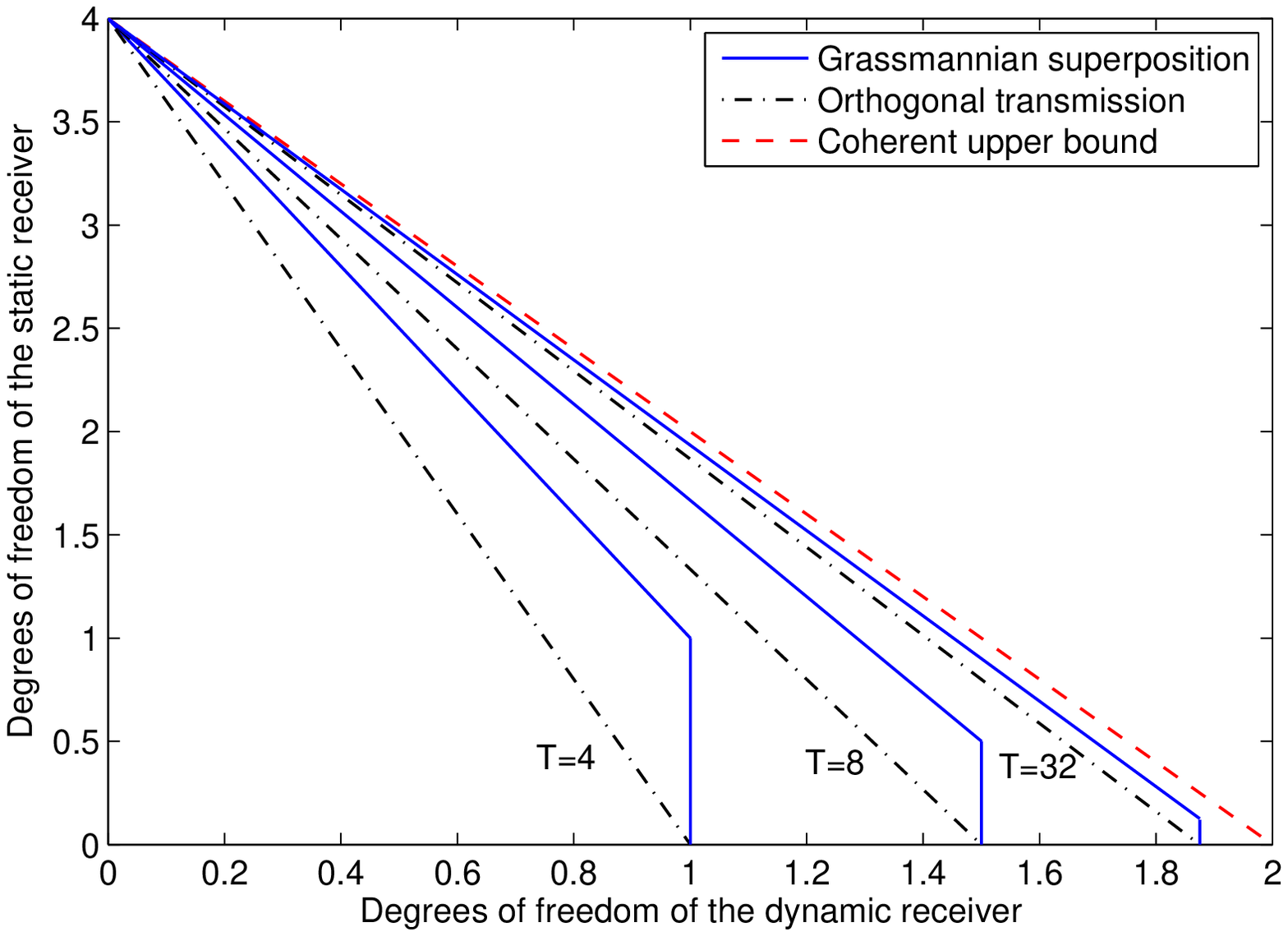}
\caption{DoF region (Theorem~\ref{thm:G-G}): $\Nn=2,\,\Nc=4$.}
\label{fig:GrassmannianDoF}
\end{figure}

We plot the achievable DoF region of Theorem~\ref{thm:G-G} in
Figure~\ref{fig:GrassmannianDoF}. For small $\Tc$, the DoF gain
achieved by the proposed method is significant, while as $\Tc$
increases, both methods approach the coherent upper
bound~\cite{Huang2009} where both of the receivers have CSIR. For
$\Tc\rightarrow \infty$, the rate gain $\Delta R_1 = O(1)$, and no DoF
gain is obtained. In this case, the achievable DoF region in
Theorem~\ref{thm:G-G} coincides with that attained by orthogonal
transmission as well as the coherent outer
bound~\cite{Huang2009}. This is not surprising, since if the channel
remains constant ($\Tc\rightarrow \infty$), the resource used for
obtaining CSIR is negligible. Finally, the rate gain $\Delta R_1$ is
an increasing function of $(\Nc-\Nn)$, i.e., the extra antennas
available for the static receiver.

Now, we design the dimension of $\Xn$ and $\Xc$ in~\eqref{eq:Tx1} to
maximize the achievable DoF region. To find the optimal dimensions, we
allow the signaling to use a flexible number of antennas and time
slots, up to the maximum available. Let $\Xn\in
\mathcal{C}^{\hat{\Nn}\times \hat{T}}$ and $\Xc\in
\mathcal{C}^{\hat{\Nc}\times \hat{\Nn}}$, where $\hat{T}\le T$,
$\hat{\Nn}\le \Nn$ and $\hat{\Nc}\le \Nc$.  Theorem~\ref{thm:G-G} does
not immediately reveal the optimal values of $\hat{\Nn}$, $\hat{\Nc}$,
and $\hat{T}$, because the rates are not monotonic in the mentioned
parameters.  The following corollary presents the optimal value of
$\hat{\Nn}$, $\hat{\Nc}$ and $\hat{T}$.

\begin{corollary}
\label{cor:dof1}
For the Grassmannian superposition under $\Nn < \Nc$, the signal dimension $\hat{T}= T$,
$\hat{\Nn}= \Nn$ and $\hat{\Nc}= \Nc$ optimizes the achievable DoF
region.
\end{corollary}
\begin{proof}
See Appendix~\ref{app:cordof1}.
\end{proof}

Thus, in the special case of $\Nn < \Nc$, it is optimal to use all
time slots and all antennas. 

\subsubsection{$\Nn\ge \Nc$}

In this case, we shall see that sometimes the Grassmanian
superposition may still outperform orthogonal transmission, but also
under certain conditions (e.g. very large $T$ or $\Nn \gg\Nc$) the
Grassmannian superposition as described in this section may be not
improve the DoF compared with orthogonal transmission.

When $\Nn\ge \Nc$, if the Grassmannian signaling to the dynamic
receiver uses all the $\Nn$ dimensions, there will remain no room for
communication with the static receiver.
To allow the static user to also use the channel, the dynamic user
must ``back off'' from using all the rate available to it, in other
words, the dimensionality of the signaling for the dynamic receiver
must be reduced. The largest value of $\hat\Nn$ that makes $\hat{\Nn}
< \Nc$ and thus allows nontrivial Grassmannian superposition is
$\hat\Nn=\Nc-1$. Once we are in this regime, the results of the
subsection~\ref{section:NnLessNc} can be used. Specifically,
Corollary~\ref{cor:dof1} indicates that de-activating any further
dynamic user antennas will not improve the DoF region. Thus, given
$\Nc$, and assuming we wish to have a non-trivial Grassmannian
signaling for both users, using $\hat{\Nn} = \Nc-1$ dimensions for
signaling to the dynamic receiver maximizes the DoF region. The
transmit signal is then
\begin{equation}
\Tx= \sqrt{\frac{\Tc}{\Nn}} \Xc\Xn,
\end{equation}
where $\Xn\in\mathcal{C}^{(\Nc-1)\times \Tc}$ and
$\Xc\in\mathcal{C}^{\Nc\times (\Nc-1)}$. The corresponding achievable
DoF pair is
\begin{align}
\mathcal{D}_4 & = \bigg((\Nc -1)(1-\frac{\Nc -1}{\Tc}),\, (\Nc -1)/\Tc \bigg),
\end{align}
which leads to the following result.

\begin{corollary}[$\Nn\ge\Nc$]
\label{cor:G-G}
Consider an $M$-antenna transmitter broadcasting to a dynamic receiver
and a static receiver with $\Nn$ and $\Nc$ antennas, respectively,
with coherence time $\Tc$ for the dynamic channel. Then the
Grassmannian superposition achieves the rate pair
\begin{align*}
\begin{cases}
\Rnc = (\Nc-1)\big(1-\frac{\Nc-1}{\Tc}\big)\log\rho + O(1)
\\ \Rc = \frac{\Nc-1}{\Tc} 
\log\rho + O(1)
\end{cases}.
\end{align*}
Denote the corresponding DoF pair with ${\mathcal
  D}_4$. Together with the two single-user operating points $\mathcal{D}_1$ and
$\mathcal{D}_2$ obtained earlier, the achievable DoF
region consists of the convex hull of $\mathcal{D}_1$, $\mathcal{D}_2$
and $\mathcal{D}_4$.
\end{corollary}
\begin{proof}
The proof follows directly by replacing $\Nn$ with $(\Nc-1)$ in
Theorem~\ref{thm:G-G}.
\end{proof}

In Corollary~\ref{cor:G-G}, the DoF for the static receiver has not
been achieved for ``free'' but at the expense of reducing the DoF
for the dynamic receiver. The transmitter uses only $\Nc-1$ dimensions
for the dynamic receiver, which allows an extra DoF $(\Nc-1)/\Tc$ to
be attained at the static receiver.
If $\Nn-\Nc$ and $\Tc$ are small, then the DoF gain of the static
receiver outweighs the DoF loss for the dynamic, so that the overall
achievable DoF region will be superior to that of orthogonal
transmission. In contrast, if $\Nn \gg \Nc$ or $\Tc$ is large, the DoF
loss from the dynamic receiver may not be compensated by the DoF gain
from the static receiver, as illustrated by
Figure~\ref{fig:GrassmannianDoF1}. Therefore in the latter case
orthogonal transmission may do better.
The following corollary specifies the condition under which
Grassmannian superposition improves DoF region compared with
orthogonal transmission.
\begin{figure}
\centering
\includegraphics[width=3.5in]{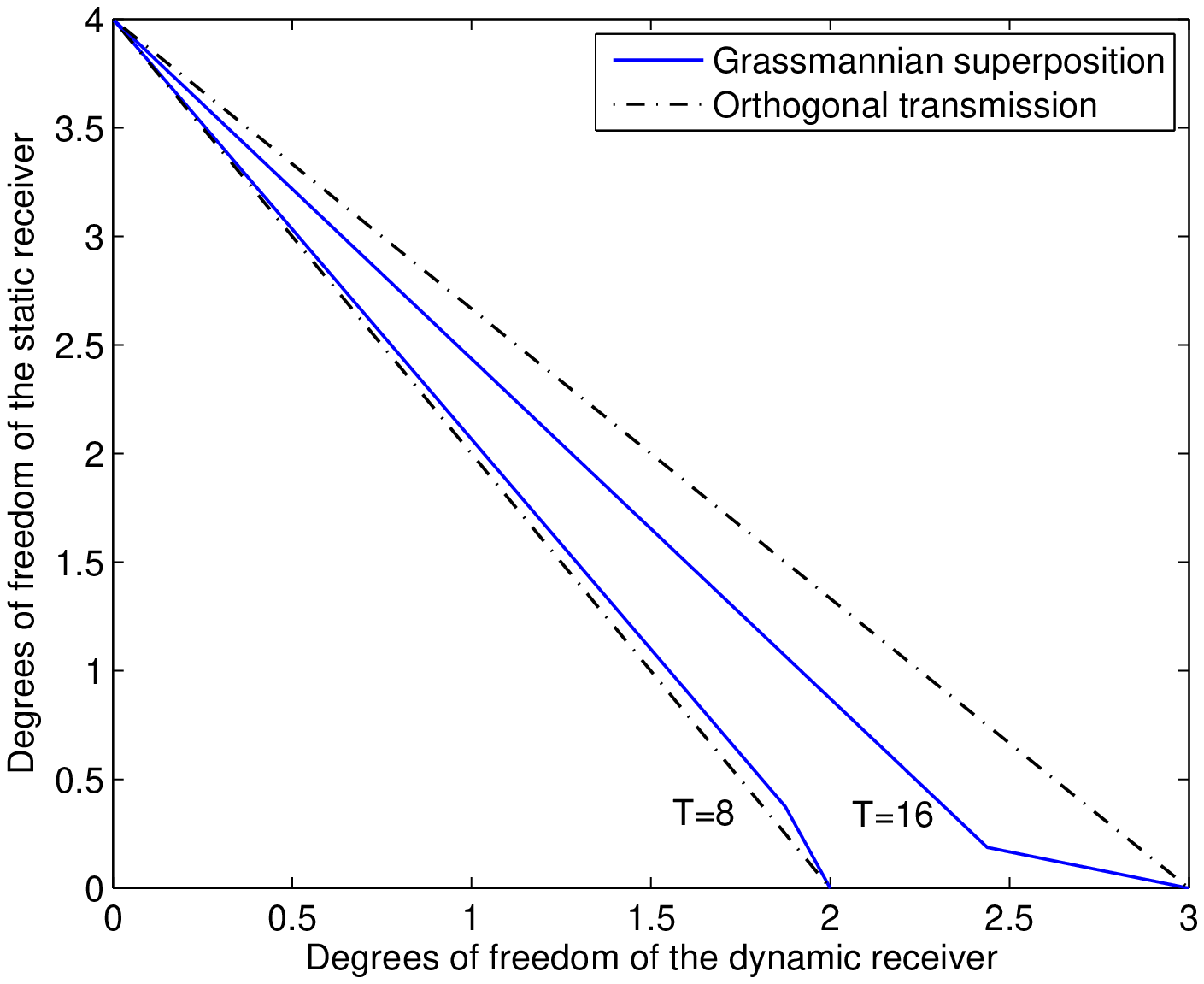}
\caption{DoF region (Corollary~\ref{cor:G-G}): $\Nn=\Nc=4$.}
\label{fig:GrassmannianDoF1}
\end{figure}

\begin{corollary}
\label{cor:dof11}
For $\Nn\ge \Nc$, the Grassmannian superposition improves DoF region
with respect to orthogonal transmission if and only if
\begin{equation}
\frac{\Nc - (\Nc -1)/\Tc}{\big(\Nc-1)(1-(\Nc-1)/\Tc\big)} < \frac{\Nc}{\Nn(1-\Nn/\Tc)}.
\end{equation}
\end{corollary}

\begin{proof}
The necessary and sufficient condition for ensuring the improvement of
the achievable DoF region is that the slope between $\mathcal{D}_2$
and $\mathcal{D}_4$ is larger than the slope between $\mathcal{D}_1$
and $\mathcal{D}_2$, which is equivalent to the inequality in the
corollary.
\end{proof}

\subsection{Design of $\mathcal{X}_1$ and $\mathcal{X}_2$}
\label{sec:codebook1}

The representation of a point in the Grassmannian is not
unique~\cite{Edelman1998} (also see Remark~\ref{remark:dimGrassmann}),
and therefore the codebooks $\mathcal{X}_1\subset\mathbb{G}(T,\Nn)$
and $\mathcal{X}_2\subset\mathbb{G}(\Nc,\Nn)$ are not unique.

First, $\mathcal{X}_2$ is chosen to be a unitary codebook. When $\Xc$
is unitary, for i.i.d. Rayleigh fading $\Hnc$, the equivalent dynamic
channel $\widetilde{\bH}_1= \Hnc \Xc$ still has i.i.d. Rayleigh fading
coefficients~\cite{Marzetta1999}. Therefore, the static receiver is
{\em transparent} to the dynamic receiver, which allows us to decouple
and simplify the design of the two codebooks and their decoders.

Once $\mathcal{X}_2$ is chosen to be a set of unitary matrices,
communication between dynamic receiver and the transmitter is
equivalent to a non-coherent point-to-point MIMO channel. Hence, to
maximize the rate of the dynamic receiver at high SNR, $\mathcal{X}_1$
must also be a collection of isotropically distributed unitary
matrices (see Section~\ref{sec:Preliminaries}).

\begin{remark}
With unitary codebooks $\mathcal{X}_1$ and $\mathcal{X}_2$,
information for both receivers is conveyed purely by the the subspace
to which the codeword belongs. Consider $\bX\in\mathcal{C}^{k\times
  n}$, $n\ge k$, which is uniquely represented by $\Omega$ (the row
space of $\bX$)  and a $k\times k$ coefficient matrix $\bC$ according
to a certain basis of $\Omega$. The codewords $\bX_1,\bX_2$ can be
represented as
\begin{align}
\bX_1 &\rightarrow (\Omega_1, \bC_1), \nonumber\\
\bX_2 &\rightarrow (\Omega_2,\bC_2).
\end{align}
In a manner similar to~\cite{Zheng2002}, one can verify
\begin{align}
I (\Xn;\Yn) & = I(\Omega_1; \Yn)+ \underset{=0}{\underbrace{ I (\bC_1;
\Yn|\Omega_1)}},
\end{align}
and
\begin{align}
I (\Xc;\Yc|\Hc)  & = I(\Omega_2; \Yc|\Hc) \nonumber \\
& \quad + \underset{=0}{\underbrace{I (\bC_2;
\Yc|\Omega_2,\Hc)}}.  
\end{align}

\end{remark}

\subsection{Multiplicative vs. Additive Superposition }
\label{sec:multiplication}

In this section, we compare product superposition with additive
superposition. Under additive superposition, the transmit signal has a
general expression
\begin{equation}
\Tx = \sqrt{c_1\rho}\, \bV_1\,\Xn + \sqrt{c_2\rho}\, \bV_2\,\Xc,
\end{equation}
where $\bV_1$ and $\bV_2$ are the precoding matrices, and $c_1$ and
$c_2$ represent the power allocation. In this case, the signal at the
dynamic receiver is
\begin{equation}
\Yn = \sqrt{c_1\rho}\, \Hnc \bV_1 \Xn + \sqrt{c_2\rho}\, \Hnc \bV_2 \Xc + \Wn.
\end{equation}
Since $\Hnc$ is unknown, the second interference term cannot be
completely eliminated in general, which leads to a bounded
signal-to-interference-plus-noise ratio (SINR), resulting in zero DoF
for the dynamic receiver. 

For the multiplicative superposition, the signal at the dynamic
receiver is
\begin{align}
\Yn & = \sqrt{c\rho}\; \Hnc\Xc\Xn + \Wn \nonumber \\ 
& = \sqrt{c\rho}\;
\widetilde{\bH}_1\Xn + \Wn,
\end{align}
where $c$ is a power normalizing constant. For any unitary $\Xc$,
$\Xc\Xn$ and $\Xn$ span the same row space. This invariant property of
Grassmannian enables us to convey information to the static receiver
via $\Xc$ without reducing the degrees of freedom of the dynamic
receiver. Intuitively, the dynamic receiver does not have CSIR and is
``insensitive'' to rotation, i.e., the distribution of $\Yn$ does not
depend on $\Xc$.

For the static receiver, the received signal is
\begin{align}
\Yc =  \sqrt{c\rho} \,\Hc \Xc\Xn +  \Wc.
\end{align}
Because $\Hc$ is known, the channel rotation $\Xc$ is detectable,
i.e., the distribution of $\Yc$ depends on $\Xc$.  Therefore $\Xc$ can
be used to convey information for the static receiver.

\section{Grassmannian-Euclidean Superposition for the Broadcast Channel}
\label{sec:Broadcasting2}

We now propose a new transmission scheme based on successive
interference cancellation, where the static receiver decodes and
removes the signal for the dynamic receiver before decoding its own
signal. This scheme improves the DoF region compared to the
non-interfering Grassmannian signaling of the previous section.

\subsection{A Toy Example}


Consider $\NTx = \Nn = \Nc=1$ and $T=2$. Our approach is that over $2$
time-slots, the transmitter sends
\begin{equation}
\bx= x_2\,
\mathbf{x}_1^t\in \mathcal{C}^{1\times 2},
\end{equation}
where $\mathbf{x}_1=[x_1^{(1)} \ x_2^{(1)} ]^t$ is the signal for the
dynamic receiver and $x_2$ is the signal for the static
receiver. Here, $\mathbf{x}_1$ has unit-norm and is from a codebook
$\mathcal{X}_1$ that is a subset of $\mathbb{G}(2,1)$, and $x_2$ can
obey any distribution that satisfies the average power
constraint.

The signal at the dynamic receiver is
\begin{align}
\mathbf{y}_1& = h_1 x_2
[x_1^{(1)}\ x_2^{(1)}] + \frac{1}{\sqrt{\rho}}[w_1^{(1)}\ w_2^{(1)}] \\ 
& = \tilde{h}_1\,[x_1^{(1)}\ x_2^{(1)}] +
\frac{1}{\sqrt{\rho}}[w_1^{(1)}\ w_2^{(1)}], \label{eq:toy21}
\end{align}
where $h_1$ is the channel coefficient of the dynamic receiver, and
$\tilde{h}_1\triangleq h_1 x_2$ is the equivalent channel
coefficient. The dynamic receiver can determine the row space spanned
by $\mathbf{x}_1$ even though $\tilde{h}_1$ is unknown,
in a manner similar to Section~\ref{sec:toy1}. The total DoF conveyed
by $\mathbf{x}_1$ is $1$ (thus $\frac{1}{2}$ per time-slot); this is
the optimal DoF under the same number of antennas and coherence time.

For the static receiver, the received signal is:
\begin{align}
\mathbf{y}_2& = h_2 x_2
[x_1^{(1)}\ x_2^{(1)}] + \frac{1}{\sqrt{\rho}}[w_1^{(2)}\ w_2^{(2)}] \\ 
& = \tilde{h}_2\,[x_1^{(1)}\ x_2^{(1)}] +
\frac{1}{\sqrt{\rho}}[w_1^{(2)}\ w_2^{(2)}], \label{eq:toy22}
\end{align}
where $h_2$ is the channel coefficient of the static receiver, and
$\tilde{h}_2 \triangleq h_2x_2$. Intuitively, since~\eqref{eq:toy21}
and~\eqref{eq:toy22} are equivalent, if the dynamic receiver decodes
the subspace of $\mathbf{x}_1$, so does the static receiver. Then, the
exact signal vector $\mathbf{x}_1$ is known to the static receiver
(recall that each subspace is uniquely represented by a signal
matrix). The static receiver removes the interference signal
$\mathbf{x}_1$
\begin{align}
\mathbf{y}_2  \mathbf{x}_1^{\dag} & = h_2 x_2 + \frac{1}{\sqrt{\rho}} \tilde{w}_2, 
\label{eq:toy23}
\end{align}
where $\tilde{w}_2$ is the equivalent noise. Finally, the static
receiver knows $h_2$, so it decodes $x_2$ and attains $1/2$ DoF per
time-slot.

Therefore, the proposed scheme attains the maximum DoF for the dynamic
receiver, meanwhile achieving $1/2$ DoF for the static receiver. With
time sharing between this scheme and $\mathcal{D}_2$, the achievable
DoF pair is
\begin{equation}
(d_1, d_2) = \big(\frac{t}{2},\, 1-\frac{t}{2}\big).
\end{equation}
Figure~\ref{fig:DoFToy2} shows that this region is uniformly larger
than that of orthogonal transmission.

\begin{figure}
\centering
\includegraphics[width=3.5in]{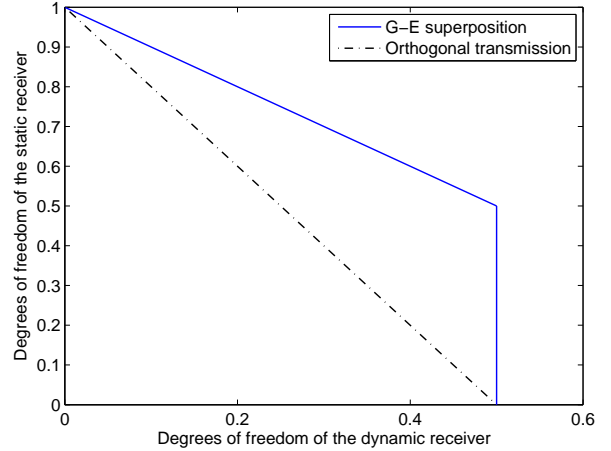}
\caption{DoF region of the toy example 2.}
\label{fig:DoFToy2}
\end{figure}

\begin{remark}
\label{remark: toyexample2}
There are two key differences between the method proposed here and the
Grassmannian superposition proposed in
Section~\ref{sec:Broadcasting1}. First, the information for the static
receiver is carried by the {\em value} of $x_2$ instead of its
direction (subspace), i.e., the signal for the static receiver is
carried in the Euclidean space. Second, the static receiver must
decode and remove the interference signal for the dynamic receiver
before decoding its own signal, which is unlike the non-interfering
method of the previous section.
\end{remark}

\subsection{Grassmannian-Euclidean Superposition Signaling}

We denote the aforementioned method as {\em Grassmannian-Euclidean
  superposition}, whose generalization is the subject of this
subsection.
Two separate cases are considered based on whether the number of
static receiver antennas is less than, or no less than, the number of
dynamic receiver antennas.

\subsubsection{$\Nn \le \Nc$}
The transmitter sends $\Tx\in \mathcal{C}^{\Nc\times \Tc}$
\begin{equation}
\Tx = \sqrt{\frac{\Tc}{\Nn\Nc}} \Xc \Xn, \label{eq:Tx21}
\end{equation}
where $\Xn\in\mathcal{C}^{\Nn\times\Tc}$ and
$\Xc\in\mathcal{C}^{\Nc\times\Nn}$ are signals for the dynamic
receiver and the static receiver, respectively. The signal $\Xn$ is
from a Grassmannian codebook $\mathcal{X}_1\subset
\mathbb{G}(\Tc,\Nn)$, while $\Xc$ is from a conventional Gaussian
codebook $\mathcal{X}_2$. The constant $\sqrt{\Tc/\Nn\Nc}$ is a
power normalizing factor.

We now give a sketch of the argument of the DoF attained by the
superposition signaling~\eqref{eq:Tx21}. For the dynamic receiver,
$\Yn\approx \Hnc \Xc \Xn$ at high SNR. When $\Nn\le\Nc$, the
equivalent channel $\Hnc \Xc \in \mathcal{C}^{\Nn\times \Nn}$ has full
rank and does not change the row space of $\Xn$. Recovering the row
space of $\Xn$ produces $(\Tc-\Nn)\Nn$ DoF, which is similar to
Section~\ref{sec:Broadcasting1}.

For the static receiver, the signal at high SNR is
\begin{equation}
\Yc\approx  \sqrt{\frac{\Tc}{\Nn\Nc}} \Hc\, \Xc\, \Xn =  \sqrt{\frac{\Tc}{\Nn\Nc}} \widetilde{\bH}_2\, \Xn.
\end{equation}
For $\Nn \le \Nc$, $\widetilde{\bH}_2=\Hc \Xc\in
\mathcal{C}^{\Nc\times\Nn}$ has full column rank and does not change
the the row space of $\Xn$, and therefore, the signal intended for the
dynamic receiver can be decoded by the static receiver.  From the
subspace spanned by $\Xn$, the codeword $\Xn\in\mathcal{X}_1$ is
identified. Then, $\Xn$ is peeled off from the static signal:
\begin{equation}
\Yc \Xn^{\dag} \approx \sqrt{\frac{\Tc}{\Nn\Nc}} \Hc \Xc \in
\mathcal{C}^{\Nc\times \Nn}. \label{eq:StaticRx2}
\end{equation}
Because $\Hc$ is known by the static receiver,
Eq.~\eqref{eq:StaticRx2} is a point-to-point MIMO
channel. Therefore, $\Nc\Nn$ DoF can be communicated via $\Xc$ to the
static receiver (over $\Tc$ time-slots)~\cite{Telatar1999}.

Altogether, the Grassmannian-Euclidean superposition attains the
DoF pair $\mathcal{D}_5$
\begin{equation}
\mathcal{D}_5 = \bigg( \Nn(1-\Nn/\Tc),\, \Nc\Nn/\Tc \bigg).
\end{equation}
More precisely, we have the following theorem.
\begin{theorem}[$\Nn\le\Nc$]
\label{thm:G-E}
Consider a broadcast channel with an $M$-antenna transmitter, a
dynamic receiver and a static receiver with $\Nn$ and $\Nc$ antennas,
respectively, with coherence time $\Tc$ for the dynamic channel. The
Grassmannian-Euclidean superposition achieves the rate pair
\begin{align*}
\begin{cases}
\Rnc = \Nn\big(1-\frac{\Nn}{\Tc}\big)\log\rho + O(1) \\ \Rc = \frac{\Nn\Nc}{\Tc} \log\rho + O(1)
\end{cases}.
\end{align*}
Denote the corresponding DoF pair by $\mathcal{D}_5$. Together with
the two single-user operating points $\mathcal{D}_1$, $\mathcal{D}_2$
obtained earlier, the achievable DoF region consists of the convex
hull of $\mathcal{D}_1$, $\mathcal{D}_2$ and $\mathcal{D}_5$.
\end{theorem}
\begin{proof}
See Appendix~\ref{app:thmdof2}.
\end{proof}

With the Grassmannian-Euclidean superposition, the static receiver
attains the following gain compared with orthogonal transmission:
\begin{equation}
 \Delta R_2=\frac{\Nn\Nc}{\Tc} \log\rho + O(1). \label{eq:R_2Gain2}
\end{equation}
\begin{figure}
\centering \includegraphics[width=3.5in]{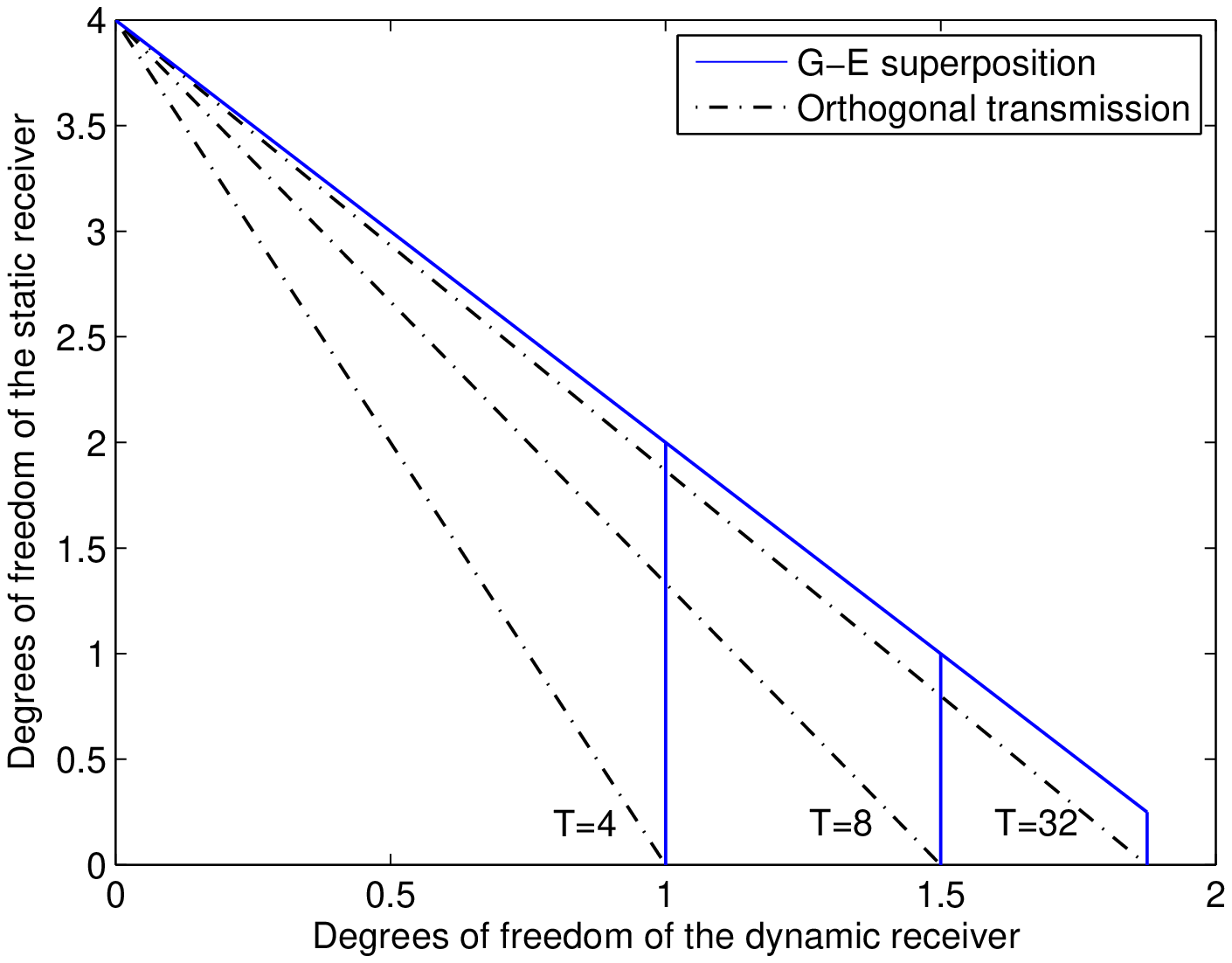}
\caption{DoF region (Theorem~\ref{thm:G-E}): $\Nn=2,\,\Nc=4$.}
\label{fig:GEDoF}
\end{figure}
From Figure~\ref{fig:GEDoF}, for relatively small $\Tc$ or large
$\Nc$, the DoF gain is significant. For example, at $\Tc=2\Nn$, the
minimum coherence interval considered in this paper, the proposed
method achieves a DoF $\Nc/2$ for the static receiver while attaining
the maximum DoF $\Nn/2$ for the dynamic receiver. As $\Tc$ increases
the gain over orthogonal transmission decreases. In the limit
$\Tc\rightarrow \infty$, we have $\Delta R_2 = O(1)$, and the DoF gain
of Grassmannian-Euclidean superposition goes away. The
Grassmannian-Euclidean superposition also provides DoF gain over the
non-interfering Grassmannian superposition\footnote{ Although
  Grassmannian-Euclidean superposition achieves larger DoF than the
  non-interfering Grassmannian superposition, it may not achieve
  larger rate at low or moderate SNR due to the decodable restriction
  on the rate (interference).}
\begin{equation}
 \Delta R = \frac{\Nn^2}{\Tc} \log\rho +
 O(1). \label{eq:R_2Gain3}
\end{equation}

The optimal design of the dimensions of $\Xn$ and $\Xc$ is trivial,
because the DoF region in Theorem~\ref{thm:G-E} is
indeed optimal (see Section~\ref{sec:dofregion}).

\subsubsection{$\Nn > \Nc$}
When the static receiver has fewer antennas than the dynamic receiver,
it may not be able to decode the dynamic signal. Here, we cannot
directly apply the signaling structure given by~\eqref{eq:Tx21}. A
straightforward way is to activate only $\Nc$ antennas at the
transmitter and use only $\Nc$ dimensions for the dynamic
receiver, that is
\begin{equation}
\Tx = \sqrt{\frac{\Tc}{\Nc^2}} \Xc \Xn
\in\mathcal{C}^{\Nc\times\Tc}, \label{eq:Tx22}
\end{equation}
where $\Xn\in\mathcal{C}^{\Nc\times T}$ and
$\Xc\in\mathcal{C}^{\Nc\times \Nc}$, and $\sqrt{\Tc/\Nc^2}$ is a power
normalizing factor. 

Following the same argument as the case of $\Nn\le \Nc$, the
Grassmannian-Euclidean superposition achieves the DoF pair
\begin{equation}
\mathcal{D}_6 = \bigg( \Nc(1-\frac{\Nc}{\Tc}),\, \frac{\Nc^2}{\Tc} \bigg).
\end{equation}

\begin{corollary}[$\Nn>\Nc$]
\label{cor:G-E}
Consider a broadcast channel with an $M$-antenna transmitter, a
dynamic receiver and a static receiver with $\Nn$ and $\Nc$ antennas,
respectively, with coherence time $\Tc$ for the dynamic channel.  The
Grassmannian-Euclidean superposition achieves the rate pair
\begin{align*}
\begin{cases}
\Rnc = \Nc\big(1-\frac{\Nc}{\Tc}\big)\log\rho + O(1)
\\ \Rc = \frac{\Nc^2}{\Tc}
\log\rho + O(1)
\end{cases}
\end{align*}
Denote the corresponding DoF pair with $\mathcal{D}_6$. Together with
the two single-user operating points $\mathcal{D}_1$ and
$\mathcal{D}_2$ obtained earlier, the achievable DoF region consists
of the convex hull of $\mathcal{D}_1$, $\mathcal{D}_2$ and
$\mathcal{D}_6$.
\end{corollary}
\begin{proof}
The proof directly follows from Theorem~\ref{thm:G-E}.
\end{proof}

In Corollary~\ref{cor:G-E}, the static rate receiver is obtained at
the expense of a reduction in the dynamic rate. The transmitter uses
only $\Nc$ out of $\Nn$ dimensions available for the dynamic receiver,
which allows extra DoF $\Nc^2/\Tc$ for the static receiver.
A necessary and sufficient condition for Grassmannian-Euclidean
superposition to improve the DoF region is as follows.

\begin{corollary}
\label{cor:dof3}
For the Grassmannian-Euclidean superposition, the signal dimension
$\hat{T}= T$, $\hat{\Nn}= \Nn$ and $\hat{\Nc}= \Nc$ optimizes the rate region at
high SNR. Moreover, it achieves superior DoF region
compared with orthogonal transmission if and only if
\begin{equation}
\Nc > (1-\frac{\Nn}{\Tc})\Nn
\end{equation}
\end{corollary}
\begin{proof}
First, using the maximum number of static antennas ($\hat{\Nc}=\Nc$)
is optimal, because both $R_1$ and $R_2$ in Corollary~\ref{cor:G-E}
are increasing functions of $\Nc$ (note that $\Nc\le \Tc/2$).

Second, we find the optimal $\hat{T}$.  Maximizing the achievable DoF
region is equivalent to maximizing the slope of the line between
$\mathcal{D}_2$ and $\mathcal{D}_4$, i.e.,
\begin{equation}
(0,\,\Nc)\quad \text{and}\quad 
\big(\Nc(1-\frac{\Nc}{\hat{T}}),\,\frac{\Nc^2}{\hat{T}}\big),
\end{equation}
which has a constant slope $-1$ and is independent of
$\hat{T}$. Therefore, any choice of $\hat{T}$, as long as $\hat{T}\ge 2\Nc$,
achieves a boundary point of the DoF region of the
Grassmannian-Euclidean superposition.

Finally, for the Grassmannian-Euclidean superposition to be superior
to orthogonal transmission in term of DoF, the slope of the line
between $\mathcal{D}_2$ and $\mathcal{D}_6$ must be larger than the
slope between $\mathcal{D}_1$ and $\mathcal{D}_2$, namely
\begin{equation}
\frac{\Nc}{(1-\Nn/T)\Nn} > 1.
\end{equation}
This completes the proof.
\end{proof}

Corollary~\ref{cor:dof3} can be interpreted as follows: the
Grassmannian-Euclidean superposition achieves superior DoF if and only
if the maximum DoF of the static receiver is larger than that of the
dynamic receiver.

\begin{figure}
\centering \includegraphics[width=3.5in]{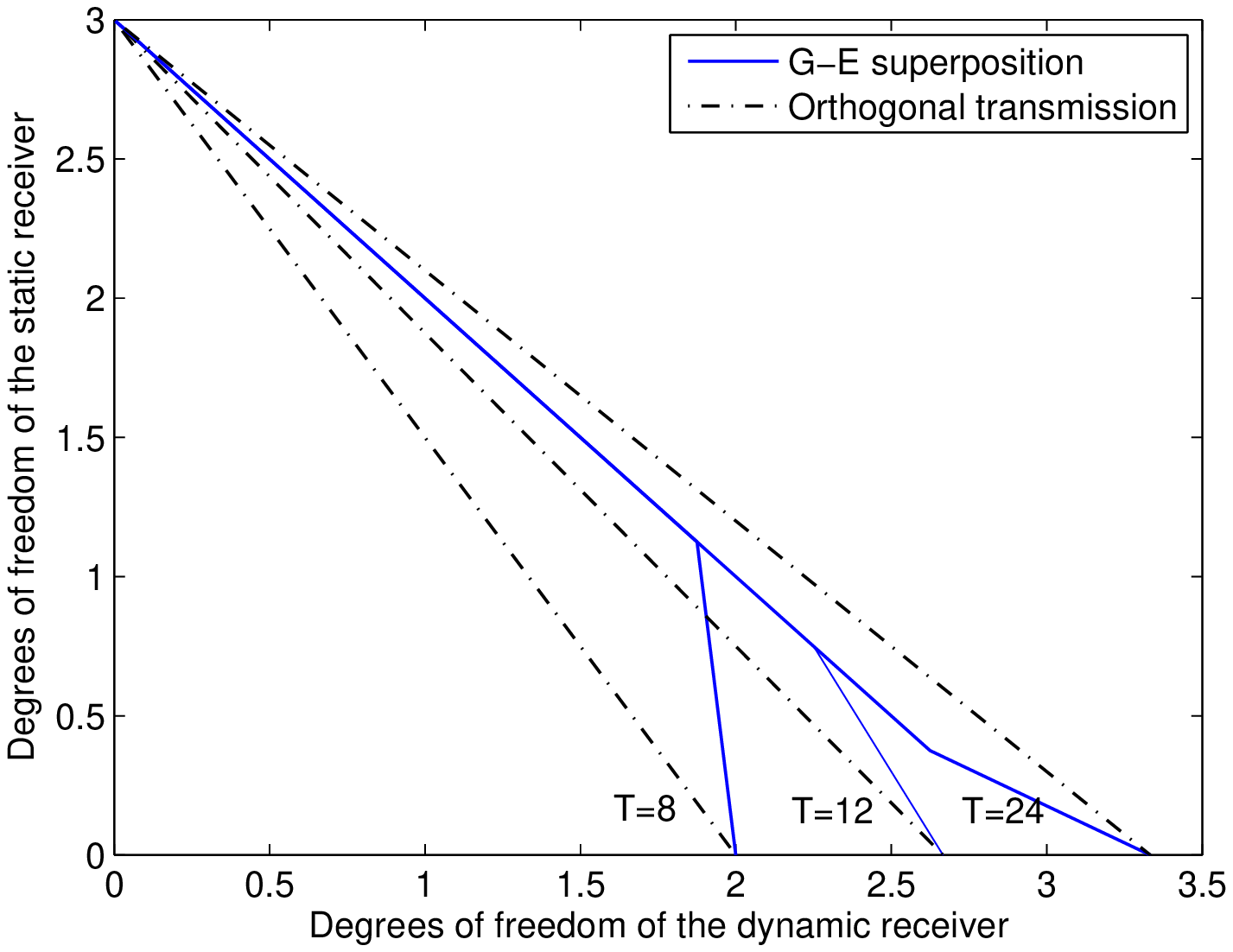}
\caption{DoF region (Corollary~\ref{cor:G-E}): $\Nn=4$,
  $\Nc=3$.}
\label{fig:GEDoF1}
\end{figure}

\subsection{Design of $\mathcal{X}_1$ and $\mathcal{X}_2$}
\label{sec:codebook2}

We heuristically argue that it is reasonable to choose $\mathcal{X}_1$
to be isotropically distributed unitary matrices and $\mathcal{X}_2$
to be i.i.d. complex Gaussian codebook. .

Recall that the Grassmannian-Euclidean superposition is to allow the
static receiver to decode the signal for the dynamic receiver and then
remove this interference. After interference cancellation, the static
receiver has an equivalent point-to-point MIMO channel with perfect
CSIR, in which case Gaussian signal achieves capacity.

Assuming $\Xc\in\mathcal{X}_2$ has i.i.d. $\mathcal{CN}(0,1)$ entries,
the equivalent channel for the dynamic receiver $\Hnc\Xc$ is
isotropically distributed (see Definition~\ref{def:id}), which leads
to two properties. First, for any $\Tc\times\Tc$ unitary matrix
$\mathbf\Phi$,
\begin{equation}
p(\Yn \mathbf{\Phi}\,|\,\Xn \mathbf{\Phi}) = p (\Yn\,|\, \Xn).
\end{equation}
Second, for any $\Nn\times\Nn$ unitary matrix $\mathbf \Psi$
\begin{equation}
p(\Yn\,|\, \mathbf{\Psi} \Xn) = p(\Yn\,|\, \Xn ).
\end{equation}
Based on these properties, the optimal signaling structure for the
channel of the dynamic receiver is a diagonal matrix\footnote{When the
  channel is i.i.d. Rayleigh fading this diagonal matrix should be
  identity at high SNR~\cite{Zheng2002}. However, it remains unknown
  whether the optimal choice is an identity matrix at arbitrary SNR.}
times a unitary matrix~\cite{Marzetta1999, Zheng2002}. Therefore,
choosing $\mathcal{X}_1$ to be isotropically distributed unitary
matrices is not far from optimal.

\subsection{Degrees of Freedom Region}
\label{sec:dofregion}

In this section, we show that the Grassmannian-Euclidean superposition
achieves the optimal DoF region under certain channel conditions.
\subsubsection{$\Nn\le \Nc$}
In this case, the optimal DoF region is as follows.
\begin{corollary}[$\Nn\le \Nc$]
\label{cor:dofregion}
When an $M$-antenna transmitter transmits to a dynamic receiver and a
static receiver with $\Nn$ and $\Nc$ antennas, respectively, with the
dynamic channel coherence time $\Tc$, the DoF region is:
\begin{align}
\begin{cases}
d_1 \le \Nn(1-\frac{\Nn}{\Tc})\\
\frac{d_1}{\Nn} + \frac{d_2}{\Nc}\le 1
\end{cases}.
\end{align}
\end{corollary}

\vspace{0.1in}

\begin{proof}
An outer bound can be found when both receivers have CSIR. The DoF
region of the coherent upper bound is~\cite{Huang2009}
\begin{equation}
\frac{d_1}{\Nn} + \frac{d_2}{\Nc} \le 1. \label{eq:dofCSIR}
\end{equation}
An inner bound is attained by Grassmannian-Euclidean superposition,
which reaches the boundary of~\eqref{eq:dofCSIR} except for $d_1>
\Nn(1-\Nn/\Tc)$. However, the DoF of the dynamic receiver can never
exceed $\Nn(1- \Nn/\Tc)$ (see
Section~\ref{sec:Preliminaries}). Therefore, Grassmannian-Euclidean
superposition achieves the DoF region.
\end{proof}

\subsubsection{ $\Nn > \Nc$}

In this case, the Grassmannian-Euclidean superposition does not match the
coherent outer bound~\eqref{eq:dofCSIR}, however, we can partially
characterize the DoF region for broadcasting with degraded message
  sets~\cite{Korner1977} and in the case of the more capable
  channel~\cite{Gamal1979}. For both cases the capacity region
is characterized by:
\begin{align}
\begin{cases}
 R_1 &\le I(\mathbf{U};\Yn)\\ R_1 + R_2 &\le I(\Xc;\Yc|\mathbf{U}) +
 I(\mathbf{U};\Yn)\\ R_1 + R_2 &\le I(\Xc;\Yc)
\end{cases},
\end{align}
where $\mathbf U$ is an auxiliary random variable. From the last
inequality we have
\begin{equation}
R_1 + R_2 \le \Nc \log \rho + O(1), 
\end{equation}
that is
\begin{equation}
d_1 + d_2 \le \Nc. \label{eq:MoreCapableBound}
\end{equation}

When $\Nc \ge (1-\Nn/\Tc) \Nn $, the inner bound in
Corollary~\ref{cor:G-E} coincides with the outer
bound~\eqref{eq:MoreCapableBound} for $0\le d_1 \le \Nc(1-\Nc/\Tc)$,
therefore, the DoF is established for this range.  For
$d_1>\Nc(1-\Nc/\Tc)$, the inner and outer bounds do not match, but the
gap is small when $\Nc$ is close to $\Nn$.

When $\Nc <(1-\Nn/\Tc) \Nn $, the inner bound in
Corollary~\ref{cor:G-E} is inferior to orthogonal transmission and the
problem remains open.

\section{Conclusion}
Signal superposition based on a multiplicative structure was proposed
to improve the degrees of freedome of the MIMO broadcast channels when
one receiver has full CSIR while the other receiver has no CSIR. Two
superposition signaling methods were proposed, both based on product
superposition. In the Grassmannian superposition, the transmit signal
is a product of two Grassmannian codewords, producing higher DoF than
orthogonal transmission while reception is still interference-free at
both receivers. The Grassmannian-Euclidean superposition uses coherent
signaling for the receiver with CSIR, and Grassmannian signaling for
the receiver without CSIR. The latter method is shown to attain the
optimal DoF region under a broad set of channel conditions.

It is possible to extend the results of this paper to more than two
receivers. The set of receivers can be divided into two sets, one with
and one without CSIR. At each point in time, the transmitter uses
product superposition to broadcast to two users, one from each
group. A scheduler selects the pair of users that is serviced at each
time. The time-sharing parameters defining the overall rate region are
as follows: one parameter determines how long a given pair is serviced
(time sharing between pairs) and for each pair a parameter determines
the operating point of the degree-of-freedom region of that pair. To
facilitate the case where there are unequal number of dynamic and
static users, the pair memberships are allowed to be non-unique, i.e.,
there may be two or more pairs that contain a given receiver. The
overall rate region is the convex hull of all rate vectors
corresponding to all values of the time-sharing parameters mentioned
above.

\appendices 

\section{Proof of Theorem~\ref{thm:G-G}}
\label{app:thmdof1}
\subsection{Achievable Rate for the Dynamic Receiver}

The normalized received signal $\Yn\in \mathcal{C}^{\Nn\times \Tc}$ at
the dynamic receiver is
\begin{equation}
\Yn = \sqrt{\frac{\Tc}{\Nn}}\, \Hnc \Xc \Xn + \frac{1}{\sqrt{\rho}}\Wn, \label{eq:Rxnc} 
\end{equation}
where $\Hnc \in \mathcal{C}^{\Nn\times \Nc}$ is the dynamic channel,
$\Xn\in\mathcal{C}^{\Nn\times \Tc}$ and $\Xc\in\mathcal{C}^{\Nc\times
  \Nn}$ are the isotropically distributed, unitary signals for the dynamic and static receivers,
respectively, and $\Wn\in \mathcal{C}^{\Nn\times \Tc}$ is additive
Gaussian noise. 

Let $\Heq1 \triangleq \Hnc\Xc$ be the $\Nn\times \Nn$ equivalent
channel, and rewrite~\eqref{eq:Rxnc} as
\begin{equation}
\Yn = \sqrt{\frac{\Tc}{\Nn}}\Heq1  \Xn +
\frac{1}{\sqrt{\rho}}\Wn.  \label{eq:EqChU1}
\end{equation}
The elements in $\Heq1 $ are
\begin{equation}
 \tilde{h}_{ij} = [\Heq1 ]_{i,j} = \sum_{k=1}^{\Nc}h_{ik}x_{kj},
 \quad 1\le i,j\le \Nn,
\end{equation}
where $h_{ik}=[\Hnc]_{ik}$ and $x_{kj}=[\Xc]_{kj}$. Note that $h_{ik}$
is i.i.d. random variable with zero mean and unit variance, therefore,
\begin{equation}
\Expt[\tilde{h}_{ij}^{\dag}\tilde{h}_{mn}] = 0 \qquad (i,j) \neq(m,n).
\end{equation}
For $(i,j)=(m,n)$ we have
\begin{align}
\Expt[|\tilde{h}_{ij}|^2] & = \sum_{k=1}^{\Nc} \Expt\big[|h_{ik}|^2
  |x_{kj}|^2 \big]  \label{eq:heq1}
\\ & = \Expt\big[\sum_{k=1}^{\Nc} |x_{kj}|^2 \big] = 1, \label{eq:heq2}
\end{align}
where~\eqref{eq:heq2} holds because $\Expt [|h_{ik}|^2]=1$ and each
column of $\Xc$ has unit norm. Therefore, the equivalent channel
$\Heq1 $ has uncorrelated entries with zero mean and unit
variance.

We now find a lower bound for the mutual information
\begin{equation}
I(\Xn;\Yn) = h(\Yn) - h(\Yn|\Xn),
\end{equation}
i.e., an achievable rate for the dynamic receiver.
First, we find an upper bound for $h(\Yn|\Xn)$. Let $\mathbf{y}_{1i}$ be
the row $i$ of $\Yn$. Using the independence bound on entropy:
\begin{equation}
h(\Yn|\Xn) \le \sum_{i=1}^{\Nn} h(\mathbf{y}_{1i}|\Xn) .  \label{eq:Y1IndependentBound1}
\end{equation}
Let $\tilde{\bh}_i$ be the row
$i$ of $\Heq1 $. Then, conditioned on $\Xn$ the covariance of
$\mathbf{y}_{1i}$ is
\begin{align}
\Expt[\mathbf{y}_{1i}^{\dag}\mathbf{y}_{1i}|\Xn] & = \frac{\Tc}{\Nn} \Xn^{\dag}\,
\Expt\big[\tilde{\bh}_i^{\dag}\tilde{\bh}_i \big]\,
\Xn + \frac{1}{\rho} \,\mathbf{I}_{\Tc}
\\ &= \frac{\Tc}{\Nn} \Xn^{\dag}\Xn +
\frac{1}{\rho} \, \mathbf{I}_{\Tc},
\end{align}
where the last equality holds since all the elements in $\Heq1 $
are uncorrelated with zero mean and unit variance. In addition, given
$\Xn$, the vector $\mathbf{y}_{1i}$ has zero mean, and therefore,
$h(\mathbf{y}_{1i}|\Xn)$ is upper bounded by the differential entropy
of a multivariate normal random vector with the same
covariance~\cite{Coverbook}:
\begin{align}
h(\mathbf{y}_{1i}|\Xn) & \le \log \det\big( \frac{\Tc}{\Nn}\Xn^{\dag}\Xn +
\frac{1}{\rho}\mathbf{I}\big) \label{eq:Y1GaussianBound1} \\
& \le \Nn \log \big(\frac{\Tc}{\Nn}  + \frac{1}{\rho}\big) - (\Tc-\Nn)\log \rho. \label{eq:Y1GaussianBound2}
\end{align}
Combining~\eqref{eq:Y1IndependentBound1}
and~\eqref{eq:Y1GaussianBound2}, we obtain
\begin{equation}
h(\Yn|\Xn) \le \Nn^2 \log \big(\frac{\Tc}{\Nn} +
\frac{1}{\rho}\big) - \Nn(\Tc-\Nn)\log
\rho. \label{eq:Y1ConditionEntropy1}
\end{equation}

After calculating the upper bound for $h(\Yn|\Xn)$, we now find a
lower bound for $h(\Yn)$ as follows.
\begin{align}
h(\Yn) & > h\big(\sqrt{\frac{\Tc}{\Nn}}\Hnc \Xc\Xn\big) \label{eq:RemoveNoise}\\ & \ge
h\big(\sqrt{\frac{\Tc}{\Nn}}\Hnc\Xc\Xn\,\big|\,\Hnc,\Xc\big),\label{eq:Rx1LB11}
\end{align}
where \eqref{eq:RemoveNoise} holds since we remove the noise, and
\eqref{eq:Rx1LB11} holds since conditioning does not increase
differential entropy. The Jacobian from $\Xn$ to $\Hnc\Xc\Xn$
is~\cite[Theorem. 2.1.5]{Muirheadbook}:
\begin{equation}
J_{X_1} = \bigg(\sqrt{\frac{\Tc}{\Nn\Nc}}\det(\Hnc\Xc)\bigg)^{\Nn}.
\end{equation}
Therefore, from~\eqref{eq:Rx1LB11} we have
\begin{align}
h(\Yn) > h(\Xn) + \Expt[\log J_{X_1}], \label{eq:Y1Entropy1}
\end{align}
where the expectation is with respect to $\Xc$ and $\Hnc$. Because
$\Xn$ is an isotropically distributed unitary matrix, i.e., uniformly
distributed on the Stiefel manifold $\mathbb{F}(\Tc,\Nn)$, we
have~\cite{Zheng2002}
\begin{equation}
h(\Xn) = \log \big|\mathbb{F}(\Tc,\Nn)\big|,
\end{equation} 
where $\big|\mathbb{F}(\Tc,\Nn)\big|$ is the volume of
$\mathbb{F}(\Tc,\Nn)$ based on the Haar measure induced by the
Lebesgue measure restricted to the Stiefel
manifold~\cite{Muirheadbook}:
\begin{equation}
\big|\mathbb{F}(\Tc,\Nn)\big| =
\prod_{i=\Tc-\Nn+1}^{\Tc}\frac{2\pi^i}{(i-1)!}. \label{eq:volumeS}
\end{equation}

Finally, combining~\eqref{eq:Y1ConditionEntropy1}
and~\eqref{eq:Y1Entropy1}, we obtain
\begin{align}
I(\Xn;\Yn) & > \Nn(\Tc-\Nn)\log \rho + \log |\mathbb{F}(\Tc, \Nn)| \nonumber 
\\ &\quad + 
\Expt[\log J_{X_1}] - \Nn\sum_{i=1}^{\Nn} \log \big(\frac{\Tc}{\Nn}  +
\frac{1}{\rho}\big)
\\ & = \Nn(\Tc-\Nn)\log \rho + O(1). \label{eq:R12}
\end{align}
Normalizing $I(\Xn;\Yn)$ over $\Tc$ time-slots yields the achievable
rate of the dynamic receiver.

\subsection{Achievable Rate for the Static Receiver}

The signal received at the static receiver is 
\begin{equation}
\Yc = \sqrt{\frac{\Tc}{\Nn}}\, \Hc \Xc \Xn +  \frac{1}{\sqrt{\rho}}\Wc, \label{eq:Rxc} 
\end{equation}
where $\Hc\in \mathcal{C}^{\Nc\times \Nc}$ is the static channel and
$\Wc\in \mathcal{C}^{\Nc\times \Tc}$ is additive Gaussian
noise. Denote the sub-matrix containing the first $\Nn$ columns of
$\Yc$ with $\Yc'$.
\begin{equation}
\Yc^{\prime} = \sqrt{\frac{\Tc}{\Nn}}\, \Hc \Xc \Xn^{\prime} +
\frac{1}{\sqrt{\rho}}\Wc^{\prime},
\end{equation}
where $\Xn^{\prime}\in\mathcal{C}^{\Nn\times\Nn}$ is
the corresponding sub-matrix of $\Xn$, and
$\Wc^{\prime}\in\mathcal{C}^{\Nn\times\Nn}$ is i.i.d. Gaussian
noise. Given $\Hc$, the mutual information between $\Yc$ and $\Xc$ is
lower bounded by:
\begin{equation}
I(\Yc;\Xc |\Hc)\ge I(\Yc^{\prime};\Xc |\Hc).
\end{equation} 
We will focus on $I(\Yc^{\prime};\Xc|\Hc)$ to derive a lower bound.
Using the singular value decomposition (SVD):
\begin{equation}
\Hc = \mathbf{U}_2^{\dag} \mathbf{\Sigma}_2 \mathbf{V}_2, \label{eq:svdHc}
\end{equation}
where $\mathbf{U}_2$, $\mathbf{V}_2 \in \mathcal{C}^{\Nc\times \Nc}$
 and
$\mathbf{\Sigma}_2 = \diag(\lambda_1,\cdots,\lambda_{\Nc})$
with $|\lambda_1|\ge \cdots\ge |\lambda_{\Nc}|$. Since
$\Hc$ is known and non-singular, the dynamic receiver applies
$\Hc^{-1}$ to remove it:
\begin{align}
\Hc^{-1}\Yc^{\prime} & = \sqrt{\frac{\Tc}{\Nn}} \Xc \Xn^{\prime} +
\frac{1}{\sqrt{\rho}}\Wc^{\prime\prime} . \label{eq:U_2Y_2}
\end{align}
The columns of $\Wc^{\prime\prime}$ are mutually independent, and each
column has an autocorrelation:
\begin{align}
\bR_W  = \mathbf{V}_2^{\dag}\,\mathbf{\Sigma}_2^{-2} \mathbf{V}_2. \label{eq:R_w}
\end{align}

Because mutual information is independent of the choice of
coordinates, we have
\begin{align}
I(\Yc^{\prime};\Xc |\Hc ) & = I(\Hc^{-1}\Yc^{\prime};\Xc|\Hc) 
\\ & = h(\Hc^{-1}\Yc^{\prime}|\Hc) -
h(\Hc^{-1}\Yc^{\prime}|\Xc,\Hc).  \label{eq:ICoordinate}
\end{align}
Let $\mathbf{y}_{2,i}$ be the column $i$ of $\Hc^{-1}\Yc^{\prime}$,
then via the independence bound on entropy:
\begin{equation}
h(\Hc^{-1}\Yc^{\prime}|\Xc,\Hc)\le \sum_{i=1}^{\Nn}
h(\mathbf{y}_{2,i}|\Xc,\Hc). \label{eq:IndependenceBound}
\end{equation}
From~\eqref{eq:U_2Y_2} and~\eqref{eq:R_w}, the
autocorrelation of $\mathbf{y}_{2,i}$ conditioned on $\Xc$ and $\bH_2$ is
\begin{align}
\bR_{2,i} & = \frac{\Tc}{\Nn}  \Xc
\Expt[\mathbf{x}_{1,i}^{\prime}\mathbf{x}_{1,i}^{\prime\dag}]\Xc^{\dag} +
\frac{1}{\rho}\bR_W
\\    & = \frac{\Tc}{\Nn} \Xc \bR_{1,i}\Xc^{\dag} + \frac{1}{\rho}\bR_W
\end{align}
where $\mathbf{x}_{1,i}^{\prime}\in\mathcal{C}^{\Nn\times 1}$ is the
column $i$ of $\Xn^{\prime}$ and has autocorrelation $\bR_{1,i}$. The
expected value of $\mathbf{y}_{2,i}$ is zero and thus the differential
entropy is maximized if $\mathbf{y}_{2,i}$ has multivariate normal
distribution~\cite{Coverbook}:
\begin{align}
 h(\mathbf{y}_{2,i}|\Xc,\Hc) & \le  \log \det \big(\frac{\Tc}{\Nn} \Xc \bR_{1,i}\Xc^{\dag} +
 \frac{1}{\rho}\bR_W\big)\nonumber
\\  & = \log \det \big(\frac{\Tc}{\Nn} \mathbf{V}_2 \Xc \bR_{1,i}\Xc^{\dag} \mathbf{V}_2^{\dag} +
 \frac{1}{\rho}\mathbf{\Sigma}_2^{-2} \big). \label{eq:DEntropyBound1}
\end{align}

The following lemma calculates $\bR_{1,i}$, the autocorrelation of a
column of an i.d. matrix.
\begin{lemma}
\label{lemma:idcorrelation}
If $\bQ\in\mathcal{C}^{\Tc\times \Tc}$ is isotropically
distributed (i.d.)  unitary matrix, then each row and column of
$\bQ$ is an i.d. unit vector 
with autocorrelation $\frac{1}{T}\bI_T$.
\end{lemma}
\begin{proof}
From Definition~\ref{def:id}, given $\bQ$ is i.d., for any
deterministic unitary matrix $\mathbf{\Phi} \in \mathcal{C}^{\Tc\times \Tc}$, 
\begin{equation}
p(\mathbf{Q\Phi}) = p(\bQ),
\end{equation}
which implies that the marginal distribution of each row and column
remains unchanged under any transform $\mathbf{\Phi}$. Therefore, each
row and column is an i.d. unit vector. Without loss of generality, we
consider the first row of $\bQ$, denoted as $\bq_1$. Let the
autocorrelation of $\bq_1$ be $\bR_q$ and posses the eigenvalue
decomposition $\bR_q = \bP^{\dag} \mathbf{\Sigma}_q \bP$, where
$\bP\in \mathcal{C}^{n\times n}$ is unitary and $\mathbf{\Sigma}_q$ is
diagonal. Since $\bq_1\bP^{\dag}$ has the same distribution as
$\bq_1$, therefore
\begin{align}
\bR_q  = \Expt[\bq_1^{\dag}\bq_1] =\bP \; \Expt[\bq_1^{\dag}\bq_1] \, \bP^{\dag} = \mathbf{\Sigma}_q.
\end{align}
Thus $\bR_q$ is a diagonal matrix. Furthermore, the diagonal elements
of $\mathbf{\Sigma}_q$ have to be identical, i.e., $\mathbf{\Sigma}_q
= a\mathbf{I}_{\Tc}$, otherwise $\bR_q$ would not be rotationally invariant
which conflicts with the i.d. assumption. Finally, because $\tr(\bR_q)
= 1$, we have $\mathbf{\Sigma}_q=\mathbf{I}_{\Tc}/\Tc$.
This completes the proof of Lemma~\ref{lemma:idcorrelation}.
\end{proof}

Since $\Xn$ is an i.d. unitary matrix, based on
Lemma~\ref{lemma:idcorrelation}, the autocorrelation of its sub-column
is
\begin{equation}
\bR_{1,i}= \mathbf{I}_{\Nn}/\Tc.
\end{equation}
Therefore, the eigenvalues of $\mathbf{V}_2 \Xc
\bR_{1,i}\Xc^{\dag} \mathbf{V}_2^{\dag}$ are
\begin{equation}
\bigg(\underset{\Nn}{\underbrace{\frac{1}{T},\cdots,\frac{1}{T}}},\underset{\Nc-\Nn}{\underbrace{0,\cdots,0}}\bigg).
\end{equation}
We now bound the eigenvalues of the sum of two matrices
in~\eqref{eq:DEntropyBound1}, noting that $\lambda_j^{-2}$ are in
ascending order and using a theorem of Weyl~\cite[Theorem 4.3.1]{Hornbook}:
\begin{align}
 h(\mathbf{y}_{2,i}|\Xc,\Hc) & \le
 \Nn\log\big(\frac{1}{\Nn} +
 \lambda_{\Nc}^{-2}\big) \nonumber \\
& \quad +(\Nc-\Nn)\log \frac{1}{\rho}\lambda_{\Nc}^{- 2}.  
\label{eq:detBound}
\end{align}

From~\eqref{eq:IndependenceBound} and~\eqref{eq:detBound}, we have:
\begin{align}
&h(\Hc^{-1}\Yc^{\prime}|\Xc,\Hc)\le \Nn^2 \log\big(\frac{1}{\Nn} +
 \lambda_{\Nc}^{-2}\big) \nonumber\\ 
& \qquad- \Nn (\Nc-\Nn)\log \lambda_{\Nc}^{-2}  - \Nn(\Nc-\Nn) \log \rho  \label{eq:hCondition}.
\end{align}

We now calculate a lower bound for $h(\Hc^{-1}\Yc^{\prime}|\Hc)$:
\begin{align}
h(\Hc^{-1}\Yc^{\prime}|\Hc) & > h(\sqrt{\frac{\Tc}{\Nn}}\Xc
\Xn^{\prime}|\Hc)\\
                & > h(\sqrt{\frac{\Tc}{\Nn}}
\Xc\Xn^{\prime} |\Xn^{\prime},\Hc). \label{eq:H_2Y_2}
\end{align}
From~\cite[Theorem. 2.1.5]{Muirheadbook}, given $\Xn^{\prime}$ the
Jacobian of the transformation from $\Xc$ to $\sqrt{\frac{\Tc}{\Nn}}
\Xc\Xn^{\prime}$ is:
\begin{equation}
J_{X_2} =
\bigg(\sqrt{\frac{\Tc}{\Nn}}\bigg)^{\Nc}\det(\Xn^{\prime})^{\Nn}.
\end{equation}
Therefore, from the right hand side of~\eqref{eq:H_2Y_2} we have
\begin{equation}
h(\Hc^{-1}\Yc^{\prime}|\Hc) > h(\Xc) + \Expt[\log J_{X_2}], \label{eq:hH_2Y_2}
\end{equation}
where the expectation is with respect to $\Xn^{\prime}$.
%
%
Because $\Xc$ is uniformly distributed on the Stiefel manifold
$\mathbb{F}(\Nc,\Nn)$, we have~\cite{Zheng2002}
\begin{equation}
h(\Xn) = \log \big|\mathbb{F}(\Nc,\Nn)\big|,
\end{equation} 
where $\big|\mathbb{F}(\Nc,\Nn)\big|$ is the volume of
$\mathbb{F}(\Nc,\Nn)$, which is given by~\cite{Muirheadbook}:
\begin{equation}
\big|\mathbb{F}(\Nc,\Nn)\big| = \prod_{i=\Nc-\Nn+1}^{\Nc}\frac{2\pi^i}{(i-1)!}.
\end{equation}

Finally, substituting~\eqref{eq:hH_2Y_2} and~\eqref{eq:hCondition}
into~\eqref{eq:ICoordinate}, we have
\begin{align}
I(\Yc^{\prime};\Xc|\Hc) & = \Nn (\Nc-\Nn) \log\rho + O(1).
\end{align}
Hence, the rate achieved by the static receiver is
\begin{align}
\frac{1}{\Tc}\Expt[I(\Yc^{\prime};\Xc|\Hc)]  = \frac{\Nn}{\Tc}(\Nc-\Nn)\log \rho
+ O(1),
\end{align}
where the expectation is with respect to $\Hc$.

\section{Proof of Corollary~\ref{cor:dof1}}
\label{app:cordof1}

The objective is to find the best dimensions for the transmit signals
$\bX_1 \in {\cal C}^{\hat{\Nn}\times\hat{T}}$ and $\bX_2 \in {\cal
  C}^{\hat{\Nc}\times\hat{\Nn}}$.
From Theorem~\ref{thm:G-G}, it is easily determined that
$\hat{\Nc}=\Nc$ is optimal, because the pre-log factor of $R_2$
increases with $\hat{\Nc}$ and the pre-log factor of $R_1$ is
independent of $\hat{\Nc}$ (given $\hat{\Nn} \le \Nc $).

To find the optimal values of $\hat{\Nn}, \hat{T}$, we start by
relaxing the variables by allowing them to be continuous valued,
i.e. $\hat{\Nn} \rightarrow x$ and $\hat{T} \rightarrow y$, and then
showing via the derivatives that the cost functions are monotonic,
therefore optimal values reside at the boundaries, which are indeed
integers. 

Using the DoF expression from 
Theorem~\ref{thm:G-G}, the slope between two achievable points
${\mathcal D}_2$ and $\mathcal{D}_3$ is:
\begin{equation}
f(x,y) = \frac{x(\Nc - x)/y-\Nc}{x(1-x/y)}.
\end{equation}
Therefore, for all $0<x\le \Nn$,
\begin{align}
\frac{\partial f(x,y)}{\partial y} & = \frac{x}{(y-x)^2}>0.
\end{align}
We wish to maximize $f$ with the constraint $y\le T$, thus $y=T$ is
optimal.

Substituting $y=T$ into $f(x,y)$, we have
\begin{align}
\frac{\partial f(x,T)}{\partial x} & = -\frac{(T-\Nc)x^2+T\Nc x -
  T^2\Nc }{x^2(T-x)^2}. \label{eq:partialx}
\end{align}

If $T=\Nc$, since $x\le T/2$, then $\frac{\partial f}{\partial
  x}>0$. In this case $x=\Nn$ maximizes the DoF region.

If $T\neq \Nc$, let $T = \alpha \Nc$. When $0<\alpha<\frac{3}{4}$, one
can verify that $\frac{\partial f}{\partial x}>0$ for all $x>0$.
Thus, $x=\Nn$ is optimal. When $\alpha\ge \frac{3}{4}$, let
$\frac{\partial f}{\partial x}=0$, and we have the corresponding
solutions:
\begin{equation}
x_{1,2} = \frac{-\alpha\Nc \pm \alpha\Nc
  \sqrt{1+4(\alpha-1)}}{2(\alpha-1)}. \label{eq:x12}
\end{equation}
When $\frac{3}{4}\le \alpha < 1$, the above solutions are positive,
where the smaller one is:
\begin{align}
x_1 = \frac{\alpha \Nc - \alpha \Nc\sqrt{1-4(1-\alpha)}}{2(1-\alpha)} > \Nn.
\end{align}
Since $\frac{\partial f}{\partial x}>0$ at $x=0$, we have
$\frac{\partial f}{\partial x}>0$ for $0\le x\le \Nn$.
When $\alpha > 1$, the (only) positive solution of~\eqref{eq:x12} is:
\begin{equation}
x_1 = \frac{\alpha \Nc + \alpha \Nc \sqrt{1+4(\alpha - 1)}} {2(\alpha - 1)} > \Nn.
\end{equation}
Once again, since $\frac{\partial f}{\partial x}>0$ at $x=0$, we have
$\frac{\partial f}{\partial x}>0$ for $0\le x\le \Nn$.

Therefore, for all cases, $x=\Nn$ maximizes the DoF
region.

\section{Proof of Theorem~\ref{thm:G-E}}
\label{app:thmdof2}

\subsection{Achievable Rate for the Dynamic Receiver}
The proof is similar to the proof for Theorem~\ref{thm:G-G}, so we
only outline key steps.
The received signal at the dynamic receiver is
\begin{equation}
\Yn = \sqrt{\frac{\Tc}{\Nn\Nc}}\Hnc \Xc\Xn +
\frac{1}{\sqrt{\rho}} \Wn,
\end{equation}
where $\Yn\in \mathcal{C}^{\Nn\times \Tc}$ and $\Hnc \in
\mathcal{C}^{\Nn\times \Nc}$ and $\Wn\in \mathcal{C}^{\Nn\times \Tc}$
is additive Gaussian noise. We establish a lower bound for the mutual
information between $\Xn$ and $\Yn$:
\begin{align}
I(\Xn;\Yn) & = h(\Yn) - h(\Yn|\Xn).
\end{align}
In the above equation, we have
\begin{align}
 h(\Yn|\Xn) &\le  \sum_{i=1}^{\Nn} h(\mathbf{y}_{1i}|\Xn) .\nonumber
\end{align}
One can verify
\begin{align}
h(\mathbf{y}_{1i}|\Xn) & \le \log \det\big( \frac{\Tc}{\Nc}\Xn^{\dag}\Xn +
\frac{1}{\rho}\mathbf{I}\big). 
\end{align}
Finally, we obtain
\begin{align}
  h(\Yn|\Xn) & < \Nn^2 \log \big(\frac{\Tc}{\Nc} +
\frac{1}{\rho}\big) - \Nn(\Tc\!-\!\Nn)\log
\rho.  \label{eq:dofR1Y1ConditionBound}
\end{align}

The lower bound is given by:
\begin{align}
h(\Yn) & >  \log |\mathbb{F}(\Tc, \Nn)| + \Expt[\log J_{X_1}] ,
\end{align}
where the expectation is with respect to $\Hnc$ and $\Xc$, and
\begin{equation}
J_{X_1} = \bigg(\sqrt{\frac{\Tc}{\Nn\Nc}}\det(\Hnc\Xc)\bigg)^{\Nn}. \label{eq:dofR1Y1Bound}
\end{equation}

Combining~\eqref{eq:dofR1Y1ConditionBound}
and~\eqref{eq:dofR1Y1Bound}, and normalizing over $\Tc$ time-slots
leads to the achievable rate of the dynamic receiver.

\subsection{Achievable Rate for the Static Receiver}
The received signal at the static receiver is $\Yc\in
\mathcal{C}^{\Nc\times \Tc}$
\begin{align}
\Yc & = \sqrt{\frac{\Tc}{\Nn\Nc}}\Hc \Xc\Xn +
\frac{1}{\sqrt{\rho}} \Wc, \nonumber
\end{align}
where $\Hc\in \mathcal{C}^{\Nc\times \Nc}$ is the static channel, and
$\Wc\in \mathcal{C}^{\Nn\times \Tc}$ is additive Gaussian noise.

We first calculate the {\em decodable} dynamic rate at the static
receiver in the next lemma.
\begin{lemma}
\label{lemma:decodable}
The static receiver is able to decode the dynamic rate $R_1$ if 
\begin{equation}
R_1 \le \Nn(1-\Nn/\Tc)\log \rho + O(1).
\end{equation}
\end{lemma}
\begin{proof}
Use the SVD for $\Hc$ and re-write the signal at the static receiver
as
\begin{align}
\Yc = \sqrt{\frac{\Tc}{\Nn\Nc}} \mathbf{U}_2^{\dag}\mathbf{\Sigma}_2\mathbf{V}_2 \Xc\Xn +
\frac{1}{\sqrt{\rho}} \Wc,
\end{align}
Because $\Xc$ is an isotropically distributed unitary matrix, $\Xc^{\prime} \triangleq
\mathbf{V}_2\Xc$ has the same distribution as $\Xc$, i.e., a matrix of
i.i.d. $\mathcal{CN}(0,1)$. Rotate $\Yc$ with $\mathbf{U}_2$
\begin{equation}
\Yc^{\prime} \triangleq \mathbf{U}_2 \Yc = \sqrt{\frac{\Tc}{\Nn\Nc}} \mathbf{\Sigma}_2 \Xc^{\prime}\Xn +
\frac{1}{\sqrt{\rho}} \Wc^{\prime},
\end{equation}
where $\Wc^{\prime}$ is i.i.d. Gaussian noise. Let
$\Yc^{\prime\prime}\in\mathcal{C}^{\Nn\times \Tc}$ be the first $\Nn$
rows of $\Yc^{\prime}$, i.e., the rows corresponding to the largest
$\Nn$ singular modes of $\Hc$, that is $|\lambda_1|\ge \cdots \ge
|\lambda_{\Nn}|$. We denote the corresponding $\Nn\times \Nn$
sub-matrix of $\Xc^{\prime}$ by $\Xc^{\prime\prime}$. Then,
\begin{equation}
\Yc^{\prime\prime} = \diag(\lambda_1,\cdots,\lambda_{\Nn})
\,\Xc^{\prime\prime} \Xn + \frac{1}{\sqrt{\rho}} \Wc^{\prime\prime}.
\end{equation}

Conditioned on $\bH_2$, the decodable dynamic rate at the static receiver is 
\begin{align}
I(\Xn;\Yc|\Hc) & = I(\Xn;\Yc'|\Hc), \nonumber 
\end{align}
which is lower bounded by
\begin{equation} 
I(\Xn;\Yc^{\prime\prime}|\Hc) = h(\Yc^{\prime\prime}|\Hc) -
h(\Yc^{\prime\prime}|\Xn, \Hc).
\end{equation}

Using the independence bound for $h(\Yc^{\prime\prime}|\Xn, \Hc)$ yields
\begin{align}
h(\Yc^{\prime\prime}|\Xn, \Hc) \le \sum_{i=1}^{\Nn}h(\mathbf{y}_{2i}|\Xn,\Hc),
\end{align}
where $\mathbf{y}_{2i}$ is the row $i$ of $\Yc^{\prime\prime}$.  Let
$\mathbf{x}_{2i}$ be the row $i$ of $\Xc^{\prime\prime}$, for $1\le
i\le\Nn$. Since $\Xc^{\prime\prime}\in\mathcal{C}^{\Nn\times \Nn}$
have i.i.d. $\mathcal{CN}(0,1)$ entries, all the row vectors
$\mathbf{x}_{2i}$ have the same autocorrelation $ I_{\Nn} $.

Conditioned on $\Xn$, the autocorrelation of $\mathbf{y}_{2i}=\lambda_i\mathbf{x}_{2i}\Xn$ is given by
\begin{align}
\Expt[\mathbf{y}_{2i}^{\dag}\mathbf{y}_{2i}|\Xn,\Hc] & =
\lambda_i^2 \Xn^{\dag} \Xn +
\frac{1}{\rho}\mathbf{I}_{\Tc}.
\end{align}
Therefore,
\begin{align}
h(\mathbf{y}_{2i}|\Xn,\Hc) &\le \log \det\big(\lambda_i^2\Xn^{\dag} \Xn +
\frac{1}{\rho}\mathbf{I}_{\Tc}\big),\\
& = \Nn \log\big(\lambda_i^2 + \frac{1}{\rho}\big) -
(\Tc-\Nn)\log \rho.
\end{align}
and subsequently, 
\begin{equation}
h(\Yc^{\prime\prime}|\Xn, \Hc) \le \sum_{i=1}^{\Nn} \log (\lambda_i^2
+\frac{1}{\rho}) - \Nn(\Tc-\Nn)\log \rho.
\end{equation}

We now find a  lower bound
for $h(\Yc^{\prime\prime} | \Hc)$. Similar to~\eqref{eq:hH_2Y_2}, we have
\begin{equation}
h(\Yc^{\prime\prime} | \Hc) \ge h(\Xn) + \Expt[J_{X_2}],
\end{equation}
where the expectation is with respect to $\Xc$, and
\begin{equation}
h(\Xn) = |\mathbb{F}(\Tc,\Nn)| =  \prod_{i=\Tc-\Nn+1}^{\Tc}\frac{2\pi^i}{(i-1)!},
\end{equation}
and 
\begin{equation}
J_{X_2} = \prod_{i=1}^{\Nn}\lambda^{2\Nn} \det(\Xc)^{\Nn}.
\end{equation}

Finally, taking expectation over $\Hc$, we obtain
\begin{align}
\Expt[I(\Xn;\Yc|\Hc)] &\ge \Nn(\Tc-\Nn)\log \rho + h(\Xn) + \Expt[J_{X_2}]
\nonumber \\ 
& \quad - \Expt\big[\sum_{i=1}^{\Nn} \log (\lambda_i^2 +\frac{1}{\rho})\big]
\nonumber \\ 
& = \Nn(\Tc-\Nn)\log \rho  + O(1).
\end{align}
This completes the proof for Lemma~\ref{lemma:decodable}.
\end{proof}

Therefore, the transmitter is able to send $\Nn(1-\Nn/\Tc)$ DoF to the
dynamic receiver, while ensuring the dynamic signal is decoded at the
static receiver.

After decoding $\Xn$, the static receiver removes the interference:
\begin{equation}
\Yc\Xn^{\dag} = \sqrt{\frac{\Tc}{\Nn\Nc}}\Hc \Xc + \frac{1}{\sqrt{\rho}}
\Wc^{\prime},
\end{equation}
where $\Wc^{\prime}\in \mathcal{C}^{\Nc\times\Nn}$ is the equivalent
noise whose entries are still i.i.d. $\mathcal{CN}(0,1)$. The
equivalent channel for the static receiver is now a  point-to-point MIMO
channel. With Gaussian input $\Xc$, we have~\cite{Telatar1999}
\begin{equation}
I(\Xc;\Yc|\Hc) = \Nn\Nc\log \rho + O(1). \label{eq:R22}
\end{equation}
Normalizing $I(\Xc;\Yc|\Hc)$ over $\Tc$ time-slots yields the
achievable rate of the static receiver.

\bibliographystyle{IEEEtran} 
\bibliography{IEEEabrv,LiYang-final}

\begin{biographynophoto}{Yang Li} (S'10) received his B.S. 
and M.S. degree in electrical engineering from Shanghai Jiao Tong
University, Shanghai, China in 2005 and 2008, respectively. He is
currently pursuing the Ph.D degree in electrical engineering at the
University of Texas at Dallas. He has interned at Samsung
Telecommunications America in 2012, and at Huawei Technologies Co. Ltd
in 2011 and 2008. His current interests include cognitive radio,
heterogeneous network, interference management and cooperative
communication.
\end{biographynophoto}

\begin{biographynophoto}{Aria Nosratinia}
(S'87-M'97-SM'04-F'10) is Jonsson Distinguished Professor of
  Engineering at the University of Texas at Dallas. He received his
  Ph.D. in Electrical and Computer Engineering from the University of
  Illinois at Urbana-Champaign in 1996. He has held visiting
  appointments at Princeton University, Rice University, and UCLA.
  His interests lie in the broad area of information theory and signal
  processing, with applications in wireless communications. He was the
  secretary of the IEEE Information Theory Society in 2010-2011 and
  the treasurer for ISIT 2010 in Austin, Texas. He has served as
  editor for the IEEE Transactions on Information Theory, IEEE
  Transactions on Wireless Communications, IEEE Signal Processing
  Letters, IEEE Transactions on Image Processing, and IEEE Wireless
  Communications (Magazine). He has been the recipient of the National
  Science Foundation career award, and is a fellow of IEEE.
\end{biographynophoto}

\end{document}